\documentclass{article}

\usepackage{arxiv}
\usepackage{setspace}
\usepackage{amsmath,amssymb,amsthm,mathtools,dsfont}
\usepackage{bm}
\usepackage{nicefrac}

\usepackage{graphicx}
\usepackage{booktabs}
\usepackage{array}
\usepackage{caption}

\usepackage{natbib}
\usepackage{hyperref}

\numberwithin{equation}{section}
\theoremstyle{plain}
\newtheorem{thm}{Theorem}[section]
\newtheorem{cor}[thm]{Corollary}
\newtheorem{lem}[thm]{Lemma}
\newtheorem{prop}[thm]{Proposition}
\theoremstyle{definition}
\newtheorem{defn}{Definition}[section]
\theoremstyle{remark}

\newcommand{\bs}{{\bf s}}

\newcommand{\bSigma}{{\boldsymbol \Sigma}}

\newcommand{\bPsi}{{\boldsymbol \Psi}}

\newcommand{\cE}{\mathcal{E}}

\newcommand{\cN}{\mathcal{N}}

\newcommand{\cU}{\mathcal{U}}

\def\cov{\hbox{Cov}}

\def\trans{^{\rm T}}

\usepackage{scalerel,stackengine}
\stackMath
\newcommand\reallywidehat[1]{%
\savestack{\tmpbox}{\stretchto{%
  \scaleto{%
    \scalerel*[\widthof{\ensuremath{#1}}]{\kern-.6pt\bigwedge\kern-.6pt}%
    {\rule[-\textheight/2]{1ex}{\textheight}}%WIDTH-LIMITED BIG WEDGE
  }{\textheight}% 
}{0.5ex}}%
\stackon[1pt]{#1}{\tmpbox}%
}

\title{Scalable Changepoint Detection for Large Spatiotemporal Data on the Sphere}

\author{
 Samantha Shi-Jun \\
  Department of Statistics\\
  University of Illinois at Urbana-Champaign\\
  Urbana, IL 61820 \\
  \texttt{wooimj2@illinois.edu} \\
   \And
 Bo Li \\
  Department of Statistics\\
  University of Washington in St. Louis\\
  St. Louis, MO 63130 \\
  \texttt{bol@wustl.edu} \\
}

\begin{document}
\maketitle
\begin{abstract}
We propose a novel Bayesian framework for changepoint detection in large-scale spherical spatiotemporal data, with broad applicability in environmental and climate sciences. Our approach models changepoints as spatially dependent categorical variables using a multinomial probit model (MPM) with a latent Gaussian process, effectively capturing complex spatial correlation structures on the sphere. To handle the high dimensionality inherent in global datasets, we leverage stochastic partial differential equations (SPDE) and spherical harmonic transformations for efficient representation and scalable inference, drastically reducing computational burden while maintaining high accuracy. Through extensive simulation studies, we demonstrate the efficiency and robustness of the proposed method for changepoint estimation, as well as the significant computational gains achieved through the combined use of the MPM and truncated spectral representations of latent processes. Finally, we apply our method to global aerosol optical depth data, successfully identifying changepoints associated with a major atmospheric event.
\end{abstract}


\section{Introduction}
Spatiotemporal datasets often exhibit abrupt shifts driven by external influences such as environmental changes, policy interventions, or natural disasters. These shifts may be manifested in various characteristics of the data, including the mean, variance, or spatial or temporal dependence structure. Identifying the time at which these changepoints occur is crucial to understanding how complex systems evolve and respond to perturbations, with applications spanning climate science, economics, and public health  \citep{shi2025tracing, berchuck2019spatially, berrett2023bayesian}.

For univariate time series, changepoint detection has been studied extensively, with comprehensive reviews provided in \cite{reeves2007review} and \cite{aminikhanghahi2017survey}. Notable methods include the pruned exact linear time (PELT) algorithm \citep{killick2012optimal} and the product partition model (PPM) \citep{barry1992product}. 
For spatiotemporal data, \cite{majumdar2005spatio}  developed a Bayesian framework for detecting various types of changepoints over time. Another related work is by \cite{xuan2007modeling}, who extended the PPM to account for dependence structures across multivariate time series using sparse Gaussian graphical models.
Both approaches, however, assume a common changepoint across all spatial locations or multivariate series, thereby enforcing that changes occur simultaneously.
Such an assumption becomes increasingly unrealistic for large-scale spatiotemporal data, especially when the change-inducing event propagates across space with a temporal lag.
 % at finer temporal resolutions, where the effect of a change-inducing event may propagate across space with a temporal lag.

Recent work addresses this limitation by modeling changepoints as a spatially correlated process, allowing them to vary across space to capture localized changes \citep{shi2025tracing, berchuck2019spatially, wang2023asynchronous}. This added flexibility, however, comes at a substantial computational cost, limiting their applicability to relatively small-scale problems. In many applications of interest, however, changepoint analysis is often required for large-scale datasets. A canonical example is global climate reanalysis products commonly used in climate science, which are observed on high-resolution grids covering the Earth’s surface with native spatial resolutions typically ranging from $0.25^\circ \times 0.25^\circ$ to $1^\circ \times 1^\circ$ \citep[e.g.,][]{hersbach2020era5, gelaro2017modern}. Addressing changepoint detection in such settings motivates the need for methods that can operate efficiently on high-resolution spherical data. 

In this paper, we focus on the problem of detecting at most one mean shift in time at each spatial location. Let $Y(\bs, t)$ denote a spatiotemporal process observed over the spatial domain $\mathbb{S}^2$ and temporal window $[0, M]$.  Following  \citet{shi2025tracing} and \citet{berchuck2019spatially}, we model $Y(\bs,t)$ as
\begin{align} \label{eq:model}
Y(\mathbf{s},t) &= \begin{cases} \mu_1(\mathbf{s},t) + U(\mathbf{s},t) + \epsilon(\mathbf{s},t), &t \leq \tau(\mathbf{s}) \\ \mu_2(\mathbf{s},t) + U(\mathbf{s},t) + \epsilon(\mathbf{s},t), &t > \tau(\mathbf{s}), \end{cases}
\end{align}
where $\tau(\mathbf{s})$ is the changepoint at location $\bf s$, $\mu_1(\mathbf{s},t)$ and $\mu_2(\mathbf{s},t)$ are mean functions before and after changepoint,  $U(\mathbf{s},t)$ is a zero mean spatiotemporal error process, and   $\epsilon(\bs,t) \overset{\mathrm{iid}}{\sim} N(0,\sigma^2_\epsilon)$ are measurement errors. 
Due to the complexity of the model structure and the number of unknown parameters, inference for models of the form \eqref{eq:model} is commonly carried out within a Bayesian hierarchical modeling (BHM) framework and fitted using Markov Chain Monte Carlo (MCMC). In this setting, the changepoint $\tau(\bs)$ is typically modeled as a spatial process to induce dependence across locations. However, non-conjugacy between the likelihood and changepoint priors necessitates Markov chain Monte Carlo (MCMC) to rely on high-dimensional Metropolis–Hastings updates to sample $\tau(\bs)$, which often suffers from poor mixing and high computational cost \citep{shand2018spatially}. Moreover, the large covariance matrices associated with both $\tau(\bs)$ and the spatiotemporal process $U(\bs,t)$ further exacerbate computational challenges as the spatial resolution increases.
% Bayesian hierarchical modeling framework (BHM) and  Markov chain Monte Carlo (MCMC) are commonly employed to fit such changepoint models. A BHM typically consists of three levels: the first layer specifies the likelihood of the observed data, followed by the second layer defining the prior models for latent processes and the third layer placing priors on hyperparameters. In our setting, Model~\eqref{eq:model} belongs to the first level, while $\tau(\bs)$ and $U(\bs,t)$ are specified at the second. The mean functions $\mu_1(\bs,t)$ and $\mu_2(\bs,t)$ can be placed in either the second or third level, depending on their structural form.  
% The changepoint $\tau(\bs)$ is typically modeled as a spatial process to induce dependence across locations. However, non-conjugacy between the likelihood and changepoint priors necessitates MCMC to rely on high-dimensional Metropolis–Hastings updates to sample $\tau(\bs)$, which often suffers from poor mixing and high computational cost \citep{shand2018spatially}. Moreover, the large covariance matrices associated with both $\tau(\bs)$ and the spatiotemporal process $U(\bs,t)$ further exacerbate computational challenges as the spatial resolution increases.

A common strategy to mitigate this computational burden is to employ reduced-rank methods, such as Fixed Rank Kriging \citep{cressie2008fixed} or multiresolution models \citep{nychka2002multiresolution}, which approximate the spatial process using a limited set of basis functions. However, these methods require determining the number of basis functions and selecting an appropriate type of basis. Choosing too few basis functions can lead to poor representation of spatial variability, while using too many may negate computational gains. Additionally, the choice of basis, whether fixed (e.g., splines, wavelets) or data-adaptive (e.g., empirical orthogonal functions), can heavily influence model performance and interpretability. Spectral methods \citep{royle2005efficient, paciorek2007computational} circumvents the need for such selections by using the Fourier basis in combination with the fast Fourier transform (FFT). This allows the use of a complete set of basis functions -- equal in number to the data points on a regular grid -- without compromising computational efficiency. However, extending this approach to spherical data introduces additional challenges. First, standard covariance functions defined in Euclidean space may not be positive definite on the sphere, which limits the class of admissible covariance functions necessitates the use of specialized kernels that respect spherical geometry \citep{guinness2016isotropic}. Furthermore, unlike the Fourier basis in the Euclidean setting, the matrix of basis functions evaluated at sampled grid points for spherical data is no longer orthogonal, which complicates both the expression and computation of spectral quantities in matrix form.

To address these challenges, we propose a scalable methodology for detecting and estimating spatially varying changepoints on a sphere.
%Our method is designed to scale efficiently with high spatial resolution, while assuming a relatively short temporal duration, making it particularly well-suited for large spatial domains with moderate time lengths. 
Our contributions are twofold. First, we introduce a novel use of the spatial multinomial probit model as a prior for the changepoints, which enables conjugacy and thus allows changepoints to be sampled efficiently via Gibbs sampler while preserving the spatial correlation structure. This modeling strategy, to our knowledge, is new in the context of spatial changepoint detection and can also be incorporated into existing frameworks such as \cite{shi2025tracing} and \cite{berrett2023bayesian} to enhance their scalability for large spatial datasets. Second, we model the spatial (and spatiotemporal) process observed on the unit sphere $\mathbb{S}^2$ as the solution to a stochastic partial differential equation (SPDE), which enables the use of eigenfunctions of the Laplace–Beltrami operator on $\mathbb{S}^2$ to alleviate the computational burden associated with large spatial covariance matrices and yields a Matérn-like covariance structure. While the SPDE approach for modeling spatial processes on manifolds, including the sphere, is well established in the literature \citep[e.g.,][]{lang2015isotropic, lindgren2011explicit, solin2020hilbert}, our contribution lies in developing a framework that integrates fast spherical harmonic transforms into fully Bayesian inference for spatiotemporal processes with measurement error, enabling efficient posterior sampling. Although developed in the context of changepoint detection, the proposed framework is broadly applicable to Bayesian modeling of large spatial and spatiotemporal datasets on the sphere.

The remainder of this paper is organized as follows. Section~\ref{sec:method} introduces the proposed method and discusses how it addresses the computational challenges outlined above. Section~\ref{sec:simulation} evaluates the effectiveness of our method in changepoint estimation and examines the trade-off between estimation accuracy and computational efficiency. In Section~\ref{sec:data}, we apply our method to global aerosol optical depth (AOD) data on a $180 \times 360$ spatial grid surrounding the time of the Mt. Pinatubo eruption. Finally, Section~\ref{sec:disc} summarizes our contributions and outlines directions for future work.
 
\section{Method}\label{sec:method}
We first introduce a prior for the changepoint process $\tau(\bs)$ that enables conjugate sampling within a Bayesian framework. We then present a SPDE-based modeling framework for spatial and spatiotemporal processes on the sphere that facilitates fast computation through spectral methods. 

\subsection{Multinomial probit prior for changepoints}\label{sec:cpmodel}
Following Model~\eqref{eq:model}, suppose data is observed over spatial locations $\bs_1, \ldots, \bs_N$ and time points $t \in \{1,2,\ldots, M\}$. Let $\mathbf{Y}_t = (Y(\bs_1,t),\ldots,Y(\bs_N,t))\trans$ and define $\mathbf{U}_t$ analogously. Define $\boldsymbol{\mu}_t(\boldsymbol{\tau}) = (\mu_i(\bs_1,t),\ldots, \mu_i(\bs_N,t))\trans$ for $i\in \{1,2\}$,  where $i=1$ if $t\leq \tau(\bs)$ and $i=2$ otherwise. 

Most parameters in \eqref{eq:model} admit conjugate priors and can be sampled via a Gibbs sampler. In contrast, specifying a prior for $\boldsymbol{\tau}=(\tau(\bs_1),\ldots,\tau(\bs_N))\trans$ is substantially more challenging, as $\boldsymbol{\tau}$ enters the likelihood in a non-linear manner through the indexing of the mean functions. To induce spatial dependence in the changepoint process, existing approaches typically model $\boldsymbol{\tau}$ as a continuous-valued spatial process, most commonly through a multivariate Gaussian distribution specified via a covariance matrix. This modeling choice, however, leads to a non-conjugate likelihood and hence a posterior distribution that is difficult to sample from directly. To see this, let $p(\boldsymbol{\tau})$ denote the prior density of $\boldsymbol{\tau}=(\tau(\bs_1),\ldots, \tau(\bs_N))\trans$. The posterior density for $\boldsymbol{\tau}$ then takes the form
\begin{align*}
    [\boldsymbol{\tau} \mid \cdot] = \frac{\prod_{t=1}^{M} \exp\left\{-\lVert\mathbf{Y}_t-\boldsymbol{\mu}_t(\boldsymbol{\tau})-\mathbf{U}_t\rVert_2^2/(2\sigma^2_\epsilon)\right\}p(\boldsymbol{\tau})}{\int \prod_{t=1}^k \exp\left\{-\lVert\mathbf{Y}_t-\boldsymbol{\mu}_t(\boldsymbol{\tau})-\mathbf{U}_t\rVert_2^2/(2\sigma^2_\epsilon)\right\}p(\boldsymbol{\tau}) d\boldsymbol{\tau}}. 
\end{align*}
Because $\boldsymbol{\mu}_t(\boldsymbol{\tau})$ is a nonlinear function of $\boldsymbol{\tau}$, no choice of prior $p(\boldsymbol{\tau})$ yields a closed form posterior amenable to Gibbs sampling. As a consequence, posterior inference for $\boldsymbol{\tau}$ generally requires high-dimensional Metropolis–Hastings updates, which scale poorly with the spatial dimension $N$.

An alternative option is to treat $\tau({\bf s})$ as a discrete variable taking values in $\{1,2,\ldots, M\}$ and specify the prior using category probabilities $\pi_k({\bf s}) := \mathbb{P}(\tau({\bf s}) = k)$ for $ k = 1,\ldots, M$. In this case, the marginal posterior distribution for each $\tau(\bs_i)$ follows a categorical distribution with probabilities given by
\begin{align}\label{eq:tau_post}
    &\mathbb{P}(\tau({\bf s}_i)=k | \cdot) \propto  \prod_{t=1}^M [Y(\bs_i, t) \mid \tau(\bs_i) = k,\cdot] \times \pi_k(\bs_i)    \\
    \propto & \prod_{t=1}^k \exp\left(\frac{-(Y(\bs_i,t)-\mu_1(\bs_i,t)-U(\bs_i,t))^2}{2\sigma^2}\right)\prod_{t=k+1}^M \exp\left(\frac{-(Y(\bs_i,t)-\mu_2(\bs_i,t)-U(\bs,t))^2}{2\sigma^2}\right) \pi_k(\bs_i),  \nonumber
    % &= \prod_{t=1}^k \exp\left(\frac{-(Y_t(\bs)-\mu^1_t(\bs)-U_t(\bs))^2}{2\sigma^2_\epsilon}\right)\prod_{t=k+1}^M \exp\left(\frac{-(Y_t(\bs)-\mu^2_t(\bs)-U_t(\bs))^2}{2\sigma^2_\epsilon}\right) \pi_k(\bs). \nonumber
\end{align} 
The normalizing constant is readily obtained since the support is discrete and finite. However, this formulation treats locations independently and therefore does not encode spatial dependence across the changepoint process, necessitating an alternative mechanism for inducing spatial correlation.

To achieve conjugacy with the data likelihood while preserving spatial correlation, we propose modeling the changepoint process using a spatial multinomial probit model, which generalizes the spatial binary regression framework proposed by \cite{paciorek2007computational}. Following \cite{albert1993bayesian}, we model $\pi_k(\bs)$ as:
\begin{equation*}\label{eq:multi-probit}
 \sum_{i=1}^k \pi_i({\bf s}) = \Phi(\gamma_k - \mu_Z({\bf s}))
 \end{equation*}
where $\gamma_1 < \cdots < \gamma_{M-1}$ are threshold parameters, $\mu_Z(\bs)$ is a spatial process, and $\Phi(\cdot)$ denotes the cumulative distribution function (CDF) of a standard normal random variable. This formulation induces the following hierarchical model for $\tau(\bs)$ via a latent Gaussian process $Z(\bs)$:
\begin{align}
    &\tau({\bf s}) \mid Z({\bf s}),\gamma_1,\ldots,\gamma_{M-1} = \begin{cases} 1 & \text{if } Z({\bf s}) < \gamma_1, \\ 
    2 & \text{if }\gamma_1 < Z({\bf s}) \leq \gamma_2, \\ 
    \vdots \\ 
    M & \text{if } Z({\bf s}) > \gamma_{M-1};  \end{cases} \label{eq:cpmodel2}\\
     &Z({\bf s}) \mid \mu_Z({\bf s}) = \mu_Z(\bs) + \epsilon_Z(\bs), \quad \epsilon_Z(\bs)\overset{\mathrm{iid}}{\sim} \mathcal{N}(0,1).\label{eq:cpmodel1}
\end{align}
To capture spatial correlation in the changepoint process, we place a Gaussian prior on $\boldsymbol{\mu}_Z  = (\mu_Z(\bs_1),\ldots, \mu_Z(\bs_N))\trans$: 
\[\boldsymbol{\mu}_Z \sim \mathcal{N}(m_Z,\bSigma_Z).\] 
Under this construction, spatial correlation in $\boldsymbol{\tau}$ is propagated through $\boldsymbol{\mu}_Z$ while the hierarchical Gaussian structure yields conjugacy, allowing the changepoints to be sampled from the exact posterior distribution via a Gibbs sampler.
    
Since \eqref{eq:cpmodel2} is invariant under translating and scaling, it is necessary to impose restrictions on the parameters to ensure model identifiability. In the standard categorical regression setting where all categories are represented in the data, identifiability is typically achieved by fixing one threshold (e.g., $\gamma_1=0$) and the marginal variance of the latent Gaussian process $\boldsymbol{\mu}_Z$. In the context of changepoint detection, however, it is possible that some categories are not observed. In particular, if the set $\{\bs_i : \tau(\bs_i) \in \{1,2\}\}$ is empty, the observed data only imply that the latent variables $Z(\bs)$ exceed higher thresholds and provide no information about how far $\mu_Z(\bs)$ lies above the first threshold $\gamma_1$. Consequently, fixing $\gamma_1$ and the marginal variance alone is insufficient to identify the location of the latent mean, and an additional constraint on the mean parameter $m_Z$ is required to ensure identifiability. The choice of $m_Z$ directly influences the probability of early changepoints, since $\mathbb{P}(\tau(\bs)=1)=\Phi(\gamma_1 - \mu_Z(\bs))$ depends on the location of $\mu_Z(\bs)$ relative to $\gamma_1$. In applications where changepoints at the very beginning of the time series are implausible, a natural choice is to set $\gamma_1=0$ and $m_Z$ to a sufficiently large positive value, thereby assigning low prior probability to $\tau(\bs)=1$. This choice should be viewed as application-dependent rather than universal. Alternative but equivalent constraints, such as fixing $m_Z=0$ and shifting $\gamma_1$, are also possible.

Using $\gamma_k \mid \gamma_{1:k-1,k+1:M}\sim \text{Unif}(\gamma_{k-1},\gamma_{k+1})$ as the prior for $k = 2,\ldots, M-1$, the resulting full conditional distribution for $\gamma_k$ is
\begin{align*}
    [\gamma_k \mid \cdot] &\propto [\boldsymbol{\tau}\mid \gamma_{1:M}, \mathbf{Z}][\gamma_k\mid\gamma_{1:k-1,k+1:M}]\\
    &\propto \prod_{i=1}^N(\mathds{1}_{\tau(\bs_i)=k+1}\mathds{1}_{\gamma_k<Z(\bs_i)} +  \mathds{1}_{\tau(\bs_i)=k}\mathds{1}_{Z(\bs_i)\leq \gamma_k}) \cdot \mathds{1}_{\gamma_{k-1}<\gamma_k<\gamma_{k+1}}\\
    &\sim \text{Unif}\Big(\max\big\{\max_\bs\{Z({\bf s})\mid \tau({\bf s}) = k\},\gamma_{k-1}\big\},\min\big\{\min_\bs\{Z({\bf s})\mid \tau({\bs})=k+1\},\gamma_{k+1}\big\}\Big).
\end{align*}
The remaining full conditionals are given by
\begin{align}
    &[Z({\bf s}) \mid \cdot, \tau({\bf s}) = k]
    \propto [\tau(\bs)=k|Z(\bs),\gamma_{1:M}][Z(\bs)\mid\mu_z(\bs)] \nonumber \\
    &\phantom{[Z({\bf s}) \mid \cdot, \tau({\bf s}) = k]}\overset{\mathrm{ind}}{\sim} \mathcal{TN}(\mu_z({\bf s}),1,a=\gamma_{k-1},b=\gamma_k) \nonumber \\
    &[\boldsymbol{\mu}_Z \mid \cdot] \sim \mathcal{N}(m_Z + (\mathbf{I}+\boldsymbol{\Sigma}_Z^{-1})^{-1}(\mathbf{Z}-m_Z),(\mathbf{I}+\boldsymbol{\Sigma}_Z^{-1})^{-1}), \label{eq:muZ-post}
    \end{align}
where $\mathcal{TN}(m,s,a,b)$ denotes the truncated normal distribution with mean $m$, variance $s$, and truncation interval $[a, b]$. The full conditional for $\boldsymbol{\mu}_Z$ follows directly from normal-normal conjugacy, as commonly presented in the Bayesian literature \citep[e.g.,][]{cressie2011statistics}. 

Together, these closed-form full conditional distributions enable a fully Gibbs-sampled inference procedure for the changepoint model \eqref{eq:cpmodel2}–\eqref{eq:cpmodel1}, with spatial dependence propagated through the latent process $\boldsymbol{\mu}_Z$. After sampling $\mu_Z(\bs)$ and $\gamma_{2:M-1}$ from their full conditionals, $\tau(\bs)$ can be updated independently across locations using \eqref{eq:tau_post}, with
\[
\pi_k(\bs) = \Phi(\gamma_k-\mu_Z(\bs)) - \Phi(\gamma_{k-1}-\mu_Z(\bs)).
\]

\subsection{Spherical harmonic representation for Gaussian random fields on a sphere}\label{sec:sh}
Although the multinomial probit formulation for $\boldsymbol{\tau}$ eliminates the need for Metropolis–Hastings updates, posterior sampling of the latent process $\boldsymbol{\mu}_Z$ remains computationally challenging when the spatial dimension $N$ is large due to the dense covariance matrix $\boldsymbol{\Sigma}_Z$. In Euclidean settings, this challenge is often addressed by working in the Fourier domain, where diagonalization enables efficient computation \citep{paciorek2007computational}. For data observed on the sphere, we adopt an analogous strategy using spherical harmonics.

We first provide a brief review of Gaussian random fields on the unit sphere. A unit sphere in $\mathbb{R}^3$ is defined as 
$\mathbb{S}^2 = \{(x,y,z) \in \mathbb{R}^3 \mid \sqrt{x^2+y^2+z^2} = 1 \}.$
In spherical coordinates, this can be equivalently written as 
\[\mathbb{S}^2 = \{(r,\theta,\phi) \in \mathbb{R}^3 \mid r=1,0\leq\theta\leq\pi, 0\leq \phi < 2\pi\},\]
\noindent where $\theta$ is the polar angle (colatitude) and $\phi$ is the azimuth (longitude), respectively.  For convenience, we denote a location on the sphere by 
$\bs = (\theta,\phi)$, omitting the radial component $r$ since it is fixed at 1.

Let $E$ be a measure space with measure $\mu$, and let $F$ be a Banach space. The space of square integrable functions from $E$ to $F$ is defined as
\[L^2(E,F) = \left\{f:E \to F \mid \int_{x \in E} \lVert f(x)\rVert^2_F d\mu(x) < \infty\right\}.\]
When $F = \mathbb{R}$ or $\mathbb{C}$, we omit $F$ and write $L^2(E).$
The $L^2$ space is a Hilbert space and has a well-defined inner product:
\[\langle f,g \rangle_{L^2(E)} = \int_E f(x)g(x)d\mu(x).\]
For $E = \mathbb{S}^2$, the inner product is defined using the surface integral:
\begin{align*}
    \langle f,g \rangle_{L^2(\mathbb{S}^2)} %&= \int_{\bs \in \mathbb{S}^2} f(\bs)g(\bs)d\bs\\
    &=\int_0^{2\pi}\int_0^\pi f(\theta,\phi)g(\theta,\phi)\sin\theta d\theta d\phi.
\end{align*}
Let $(\Omega,\mathcal{F},\mathbb{P})$ be a probability space. A Gaussian random field (GRF) $X$ on the sphere is a measurable mapping $X: \Omega \to L^2(\mathbb{S}^2)$ such that the vector of random variables $(X(\bs_1),\ldots,X(\bs_n)) \in \mathbb{R}^n$ are jointly Gaussian for any collection $\bs_1, \ldots, \bs_n \in \mathbb{S}^2$, $n \in \mathbb{N}$. In this paper, we restrict attention to GRFs that are square-integrable, i.e., $X \in L^2(\Omega, L^2(\mathbb{S}^2))$, and isotropic, i.e., $(X(\bs_1),\ldots,X(\bs_n))\trans \overset{d}{=} (X(g(\bs_1)),\ldots,X(g(\bs_n)))\trans$  for all $g \in SO(3)$, where $SO(3)$ denotes the group of rotations on $\mathbb{S}^2$.
%note to self: When $Y$ is a measure space, $L^2(X,L^2(Y))$ is isomorphic to $L^2(X\times Y)$ by Fubini's theorem

Assuming that $\mu_Z(\bs)$ is a square-integrable, isotropic GRF, it admits a spectral expansion \citep{lang2015isotropic}
\begin{align}\label{eq:kl_exp}
    \mu_Z(\theta,\phi) = m_Z + \sum_{l=0}^\infty \sum_{m=-l}^l \alpha_{lm}\psi_{lm}(\theta,\phi), 
\end{align}
where $\psi_{l,m}(\theta,\phi)$ is the spherical harmonic function of degree $l$ and order $m$ evaluated at location $(\theta,\phi) \in \mathbb{S}^2$ and $\alpha_{lm} = \langle \mu_z,\psi_{lm} \rangle_{L^2(\mathbb{S}^2)}$ are independent Gaussian random variables with $\alpha_{lm} \sim N(0,S_l),$ where $S_l$ is the angular power spectrum of $\mu_Z.$
The values for $\{S_l\}_{l=0,1,\ldots}$ determine the covariance structure of $\mu_Z$ and thus the covariance matrix $\boldsymbol{\Sigma}_Z$. 
% However, estimating an infinite or a large number of spectral coefficients is impractical, especially in the absence of strong prior knowledge about their precise structure. To ensure tractability while maintaining flexibility, we opt to model the angular power spectrum $S_l$ using a parametric family of covariance functions. 

To specify $S_l$, we employ a Matérn-type covariance model via the stochastic partial differential equation (SPDE) representation. The Mat\'ern covariance function is a popular choice for modeling the covariance of isotropic spatial processes due to its flexibility in controlling the smoothness of the underlying field \citep{stein1999interpolation}. In Euclidean space, the Mat\'ern covariance between two locations separated by distance $h$ is given by 
\[C(h) = \frac{\sigma^2}{2^{\nu-1}\Gamma(\nu)}(\kappa h)^\nu K_\nu(\kappa h ),\]
where $K_\nu$ is the modified Bessel function, $\nu > 0$ is the smoothness parameter, $\kappa > 0$ is the inverse range parameter, and $\sigma^2 > 0$ is the scaling constant. In $\mathbb{R}^d$, \cite{whittle1954stationary} showed that a Matérn GRF arises as the stationary solution to 
\begin{align}\label{eq:SPDE}
    \tilde{\sigma}^{-1}(\kappa^2 - \Delta)^{(v + d/2)/2}X({\bf s}) = \mathcal{W}({\bf s}), \quad \kappa > 0, v > 0,
\end{align}
where $\tilde{\sigma}^2 = \Gamma(\nu + d/2)(4\pi)^{d/2}\kappa^{2\nu}\sigma^2/\Gamma(\nu)$, $\Delta = \partial^2/\partial x_1^2 + \ldots + \partial^2/\partial x_d^2$ is the Laplacian,  and $\mathcal{W}$ is Gaussian white noise in $\mathbb{R}^d$. This construction extends naturally to compact manifolds through the eigenfunctions of the Laplace–Beltrami operator \citep{lindgren2011explicit}. For $\mathbb{S}^2,$ the resulting spectral representation coincides with \eqref{eq:kl_exp}, with 
\begin{equation*}\label{eq:S_l}
    \alpha_{lm} \sim \mathcal{N}(0,\sigma^2_Z(\kappa^2+l(l+1))^{-(\nu+1)} ).
\end{equation*}
See Appendix~\ref{app:appendix_a} for detailed derivation.
% it can be shown (Appendix~\ref{app:appendix_a}) that the solution to \eqref{eq:SPDE} has a spectral representation given in \eqref{eq:kl_exp} with $S_l= \tilde{\sigma}^2(\kappa^2+l(l+1))^{-(\nu+1)}$. Therefore, in \eqref{eq:kl_exp} we assume
% \begin{align*}\label{eq:S_l}
%     &\alpha_{lm} \sim \mathcal{N}(0,\sigma^2_Z(\kappa^2+l(l+1))^{-(\nu+1)} ),
% \end{align*}
% and $\alpha_{lm} , \alpha_{l'm'}$ are independent for all $(l,m) \neq (l',m')$, $0 \leq |m| \leq l \leq n$ for all $n\in\mathbb{N}$.  

For practical implementation, the expansion in \eqref{eq:kl_exp} is truncated at a finite degree $L$. We define the truncated approximation to $\mu_Z$ as:
\[ \mu_Z^L(\theta,\phi) = m_Z + \sum_{l=0}^L \sum_{m=-l}^l \alpha_{lm}\psi_{lm}(\theta,\phi).\]
For $\nu>1/2$, $\mu_Z^L$ converges to $\mu_Z$ in $L^2$ as $L\to\infty$ \citep{lang2015isotropic}. However, rate of convergence of $\mu_Z^L$ itself is not of primary interest, as the process $\mu_Z$ serves as a latent auxiliary variable whose role is to induce dependence in the changepoint process via discretization rule in \eqref{eq:cpmodel2}. Rather, the relevant question is how the truncation error in $\mu_Z^L$ propagates through the multinomial probit model and affects the induced changepoint $\tau^L$. 

The following result provides probabilistic guarantees on the deviation of the truncated changepoint $\tau^L$ from $\tau$ as a function of the truncation error in the latent process $Z^L$:
\begin{thm}
\label{thm:cp_trunc}
Assume $m_Z = 0$ and $\nu > 1/2.$ Define $v_Z := \sigma^2\sum_{l=0}^L (2l+1)(\kappa^2 + l(l+1))^{-\nu + 1}$ and $\Delta_{k,a}(z)
:=
\min\left\{
\gamma_{\lfloor k+a,M \rfloor}-z,\;
z-\gamma_{\lceil k-a-1,0 \rceil}
\right\}$. The marginal distribution of the latent process $Z^L$ is given by $Z^L(\bs) \sim \mathcal{N}(0,v_z+1),$  and
\begin{align*}
\mathbb{P}\!\left(\lvert \tau^L(\bs)-\tau(\bs)\rvert \le a\right)
&\ge
\sum_{k=1}^M
\mathbb{E}\!\left[
2\Phi\!\left(
\frac{\Delta_{k,a}(Z^L(\bs))}
{\sigma_Z L^{-\nu}
\sqrt{\frac{1}{\nu}+\frac{1}{L(2\nu+1)}}
}
\right)
-1
\;\middle|\;
\gamma_{k-1}\le Z^L(\bs)\le \gamma_k
\right]
\\
&\hspace{1.2cm}\times
\left[
\Phi\!\left(\frac{\gamma_k}{\sqrt{v_Z+1}}\right)
-
\Phi\!\left(\frac{\gamma_{k-1}}{\sqrt{v_Z+1}}\right)
\right], \ \  a = 0, \ldots, M-1.
\end{align*}
\end{thm}

The bounds in Theorem~\ref{thm:cp_trunc} provide direct control over the truncation-induced error in changepoint estimation at each location and can be used to obtain upper bounds on the expected mean absolute error $\mathbb{E}(\lVert \boldsymbol{\tau}^L - \boldsymbol{\tau}\rVert_{L^1(\mathbb{S}^2)})$. Detailed proofs, along with numerical illustrations under varying truncation levels and parameter settings, are provided in Appendix~\ref{app:appendix_trunc}.

In matrix form, the truncated process can be written as
\begin{align}\label{eq:spec}
    \boldsymbol{\mu}_Z^L = m_Z + {\bf \Psi}\boldsymbol{\alpha},
\end{align}
where $\boldsymbol{\mu}_Z^L = (\mu_Z^L(\mathbf{s}_1), \ldots , \mu_Z^L(\mathbf{s}_N))^T$, ${\bf \Psi}$ is an $N \times (L+1)^2$ matrix consisting of entries $\bPsi_{ij} = \psi_{l_j m_j}(\bs_i)$ with $0 \leq l_j \leq L$ and $-l_j \leq m_j \leq l_j$, and $\boldsymbol{\alpha}$ is an $(L+1)^2$-dimensional vector of spectral coefficients. 

Although the independence of $\alpha_{lm}$ induces a diagonal prior covariance for $\boldsymbol{\alpha}$, posterior inference remains computationally prohibitive due to the cost of constructing and multiplying by $\boldsymbol{\Psi}$. To address this bottleneck, we assume observations lie on a regular spatial grid, allowing matrix operations involving $\boldsymbol{\Psi}$ to be replaced by spherical harmonic transforms. This assumption is natural for many large-scale climate datasets, which are commonly distributed on regular latitude–longitude grids. Extensions to irregularly spaced data can be achieved through interpolation or grid augmentation \citep{paciorek2007bayesian}.

While spherical harmonic transforms and their fast implementations are well established \citep[e.g.,][]{mohlenkamp1999fast, healy2003ffts}, their standard algorithmic descriptions focus on numerical evaluation and do not explicitly characterize the associated linear operators, making their direct use in Bayesian inference less straightforward. To bridge this gap, we propose the following matrix-based formulation that explicitly links spherical harmonic transforms to weighted linear mappings in a Gaussian framework.

A central ingredient of this formulation is the use of a quadrature rule that admit exact integration of spherical functions. For such quadrature rules, the inner product of two spherical harmonic functions can be expressed as
\begin{equation*}
\delta_{ll'}\delta_{mm'} =   \langle \psi_{lm}, \psi_{l'm'} \rangle_{\mathbb{S}^2} = \sum_{i=1}^{N} \psi_{lm}({\bf s}_i)\psi_{l'm'}({\bf s}_i)w({\bf s}_i),
\end{equation*}
where $w(s_i)$ denotes the quadrature weights. Writing $\mathbf{D}_w$ for the diagonal matrix of these weights, this implies the weighted discrete orthogonality condition $\boldsymbol{\Psi}\trans \mathbf{D}_w \boldsymbol{\Psi} \;=\; \mathbf{I}.$
Consequently, left multiplication by $\boldsymbol{\Psi}\trans \mathbf{D}_w$ corresponds to the discrete spherical harmonic transform, while left multiplication by $\boldsymbol{\Psi}$ yields the inverse transform. This representation allows quantities defined in the spatial domain to be expressed linearly in the spectral domain, facilitating characterization of covariance structures and conjugate posterior updates while allowing the computation to be carried out without explicitly constructing $\boldsymbol{\Psi}$ or $\mathbf{D}_w.$ 

In this paper, we adopt the Driscoll-Healy quadrature \citep{ driscoll1994computing}, which assumes a $K\times 2K$ equally spaced grid with even $K$ and admits exact quadrature for spherical harmonics up to degree $L = K/2 - 1.$ This grid algins naturally with regular latitude--longitude (``plate carr\'ee'') formats commonly used in climate datasets and is supported by fast implementations in the \texttt{SHTools} library \citep{wieczorek2018shtools}. 
% (Equation~9, \cite{driscoll1994computing}; an explicit expression can be found in Appendix C)

Let $\mathbf{D}_\alpha$ denote the diagonal matrix of spectral variances $S_l$. The full conditional distribution of $\boldsymbol{\alpha}$ is then given by
\[
\boldsymbol{\alpha} \mid \cdot \sim \mathcal{N}\left( \left( ({\bf\Psi\trans D}_w^2{\bf\Psi})^{-1} + \mathbf{D}_{\alpha}^{-1}\right)^{-1}({\bf\Psi\trans D}_w^2{\bf\Psi})^{-1}{\bf \Psi\trans D}_w({\bf Z} - m_Z), \left( ({\bf\Psi\trans D}_w^2{\bf\Psi})^{-1} + \mathbf{D}_{\alpha}^{-1}\right)^{-1}  \right).
\]
% This expression parallels the full conditional for the spatial-domain latent vector $\boldsymbol{\mu}_Z$ in \eqref{eq:muZ-post}, with the identity matrix replaced by the transform-induced term $\boldsymbol{\Psi}\trans \mathbf{D}_w^2 \boldsymbol{\Psi}$ and with $\boldsymbol{Z}-m_Z$ mapped to the spectral domain via $\boldsymbol{\Psi}\trans D_w(\boldsymbol{Z}-m_Z)$. 
At first glance, this expression appears challenging to sample from due to the involved matrix operations. However, the matrix ${\bf D}_\alpha^{-1}$ is diagonal as discussed previously, and the term ${\bf \Psi\trans D}_w({\bf Z} - m_Z)$ can be evaluated efficiently using spherical harmonic transform, avoiding explicit matrix multiplication. The remaining challenge is the matrix $\boldsymbol{\Psi}\trans \mathbf{D}_w^2 \boldsymbol{\Psi}$, which arises from the spherical transformation of the measurement error $\boldsymbol{\epsilon}_Z \sim \mathcal{N}(0,\mathbf{I})$. Unlike in Euclidean settings where iid measurement error coincides with white noise and the Fourier transform preserves independence, the spherical harmonic transform induces a structured correlation in the transformed errors. This phenomenon does not arise in standard FFT-based spectral methods (e.g., \citealp{paciorek2007computational}) and must be explicitly addressed for valid posterior inference on the sphere. Importantly, this matrix depends only on the quadrature weights and the basis functions and thus can be precomputed and reused across MCMC iterations. Moreover, we establish in the following result that $\boldsymbol{\Psi}\trans \mathbf{D}_w^2 \boldsymbol{\Psi}$ admits a sparse, block-diagonal structure, enabling efficient posterior updates.

\begin{thm}
\label{thm:sparse}
${\bf \Psi\trans  D}_w^2{\bf \Psi}$ is sparse and has $\sum_{m=-L}^L \lceil \frac{(L+1-|m|)^2}{2}\rceil$ non-zero entries and \\$\sum_{m=0}^L \lceil \frac{(L+1-|m|)^2}{2}\rceil$ distinct non-zero entries. Further, it can be arranged such that it is block diagonal.
\end{thm}

See Appendix~\ref{app:sh_proof} for the proof. This structural property is central to making the Gibbs updates practical for large $N$ and substantially improves computational efficiency in both runtime and memory usage. Specifically, whereas matrix–vector multiplication for a dense $n \times n$ matrix requires $O(n^2)$ operations, the computational cost is reduced to $O(k)$ for a sparse matrix with $k$ nonzero entries \citep{rue2005gaussian}.

\subsection{SPDE-based spherical harmonic representation for spatiotemporal processes}\label{sec:st}
Similar to the latent spatial field $\boldsymbol{\mu}_Z,$ direct sampling of the spatiotemporal process $U(\bs, t)$ in \eqref{eq:model} is computationally challenging due to the high dimensionality of the associated covariance matrix. To address this issue, we adopt a strategy analogous to that in Section~\ref{sec:sh} and construct $U(\bs,t)$ within an SPDE framework. To facilitate spectral representation, we exploit the fact that spherical harmonics form an eigenbasis of the Laplace–Beltrami operator on $\mathbb{S}^2,$ and model $U(\bs,t)$ as the solution to a stochastic reaction–diffusion equation with additive noise and initial condition $U(\cdot,0) = \cU_0 \in L^2(\Omega,L^2(\mathbb{S}^2))$: 
\begin{align}\label{eq:heat}
d \cU(t) - \xi_d \Delta \cU(t)dt + \xi_r \cU(t)dt = d\mathcal{E}_Q(t),
\end{align}
where $\cU$ is an $L^2(\Omega,L^2(\mathbb{S}^2))$-valued process defined on the interval $[0,M]$ such that $(\cU(t))(\bs) = U(\bs,t)$, and $\mathcal{E}_Q(t) = \int_0^t d\mathcal{E}_Q(t')$ is a $Q$-Wiener process \citep{da2014stochastic}. This provides a dynamic extension of the spatial SPDE model in Section~\ref{sec:sh}

A Q-Wiener process $\mathcal{E}_Q$ in $L^2(\Omega,L^2(\mathbb{S}^2))$ admits a spherical harmonic expansion
\begin{align*}
  \mathcal{E}_Q(t) = \sum_{l=1}^\infty\sum_{m=-l}^l \sqrt{q_{lm}}\beta_{lm}(t)\psi_{lm},  
\end{align*}
where $(\beta_{lm}), l\in\mathbb{N},m=-l,\ldots,l$ is a sequence of independent Brownian motions and $q_{lm}$ are the spectral variances with values determined by the precision operator $Q.$ To obtain the values for $q_{lm},$ we again consider Mat\'ern class of covariance functions and set $Q^{1/2} = \sigma_Q(\kappa^2 - \Delta)^{-(\nu+1)/2}.$ Under this specification, $\mathcal{E}_Q(t) /\sqrt{t}$ can be obtained as the solution to Whittle-Mat\'ern SPDE in \eqref{eq:SPDE}, with $q_{lm} = \sigma^2_Q(\kappa^2+l(l+1))^{-(\nu+1)}.$ 

Since $\cU(t)$ is square-integrable, it admits an expansion
\begin{align}\label{eq:kl_U}
    \cU(t) &= \sum_{l=0}^\infty \sum_{m=-l}^l \langle \cU(t),\psi_{lm} \rangle\psi_{lm}\\
    &= \sum_{l=0}^\infty \sum_{m=-l}^l \left \langle \cU_0 + \int_0^t(\xi_r -\xi_d \Delta)\cU(t')dt' + \cE_Q(t),\psi_{lm} \right \rangle\psi_{lm} \nonumber\\
    &= \sum_{l=0}^\infty \sum_{m=-l}^l\left( \langle \cU_0,\psi_{lm}\rangle + \int_0^t \left\langle (\xi_r -\xi_d \Delta)\cU(t'),\psi_{lm}\right\rangle dt' + \sqrt{q_{lm}}\beta_{lm}(t) \right)\psi_{lm} \nonumber \\
    &= \sum_{l=0}^\infty \sum_{m=-l}^l\left( \langle \cU_0,\psi_{lm}\rangle + \int_0^t (\xi_r + \xi_dl(l+1))\left\langle \cU(t'),\psi_{lm}\right\rangle dt' + \sigma_Q(\kappa^2+l(l+1))^{-\frac{\nu+1}{2}}\beta_{lm}(t) \right)\psi_{lm}, \nonumber
\end{align}
where we used the fact that negative Laplacian on $\mathbb{S}^2$ is self-adjoint with eigenfunctions $\psi_{lm}$ with corresponding eigenvalues $l(l+1)$. 

Define $\hat{\cU}_{lm}(t) := \langle \cU(t), \psi_{lm} \rangle_{L^2(\mathbb{S}^2)}.$ Since $\beta_{lm}$ are independent, solving \eqref{eq:heat} reduces to solving the following stochastic ordinary differential equations separately for each $l = 0,1,\ldots,$ and $m = -l,\ldots,l$: 
\begin{align}\label{eq:o-u}
    \hat{\cU}_{lm}(t) = \hat{\cU}^0_{lm} + (\xi_r +\xi_d l(l+1))\int_0^t\hat{\cU}_{lm}(t')dt' + \sigma_Q(\kappa^2 + l(l+1))^{-\frac{\nu+1}{2}}\beta_{lm}(t),
\end{align}
where $\hat{\cU}_{lm}^0 = \langle \cU_0,\psi_{lm}\rangle_{L^2(\mathbb{S}^2)}.$
The solution to \eqref{eq:o-u} is an Ornstein-Uhlenbeck process and can be represented as
\begin{align}\label{eq:o-u_sol}
    \hat{\cU}_{lm}(t) &= e^{-(\xi_r+\xi_d l(l+1))t}\hat{\cU}^0_{lm} + \sigma_Q(\kappa^2+l(l+1))^{-\frac{\nu+1}{2}}\int_0^t e^{-(\xi_r+\xi_d l(l+1))(t-t')}d\beta_{lm}(t').  %\\
    % &\overset{d}= e^{-(\xi_r+\xi_d l(l+1))t}\hat{\cU}^0_{lm} + \sigma_Q(\kappa^2+l(l+1))^{-\frac{\nu+1}{2}}\sqrt{\frac{(1-e^{-2(\xi_r+\xi_d l(l+1))t})}{2(\xi_r+\xi_d l(l+1))}}\epsilon_{lm}(t), \nonumber
\end{align}
For step size $h>0,$ \eqref{eq:o-u_sol} satisfies
\begin{align*}
    \hat{\cU}_{lm}(t+h) 
    % &= e^{-(\xi_r+\xi_d l(l+1))h}e^{-(\xi_r+\xi_d l(l+1))t}\hat{\cU}^0_{lm} + \sigma_Q(\kappa^2+l(l+1))^{-\frac{\nu+1}{2}}\int_0^{t+h} e^{-(\xi_r+\xi_d l(l+1))(t + h-t')}d\beta_{lm}(t')\\
    &= e^{-(\xi_r+\xi_d l(l+1))h}\hat{\cU}_{lm}(t) + \sigma_Q(\kappa^2+l(l+1))^{-\frac{\nu+1}{2}}\int_t^{t+h} e^{-(\xi_r+\xi_d l(l+1))(t + h-t')}d\beta_{lm}(t').
\end{align*}
Thus, for uniformly and discretely sampled time points $t_i = i,$ $\hat{\cU}_{lm}$ can be expressed as an AR(1) process:
\begin{align}\label{eq:diffdyn}
    \hat{\cU}_{lm}(t) \overset{d}= e^{-(\xi_r+\xi_d l(l+1))}\hat{\cU}_{lm}(t-1) + \sigma_Q(\kappa^2+l(l+1))^{-\frac{\nu+1}{2}}\sqrt{\frac{(1-e^{-2(\xi_r+\xi_d l(l+1))})}{2(\xi_r+\xi_d l(l+1))}}\epsilon_{lm}(t),
\end{align}
where $\epsilon_{lm}(t) \overset{\mathrm{iid}}{\sim} \cN(0,1)$. 

Let ${\bf U}^L_t = (U^L(s_1,t),\ldots, U^L(s_N,t))\trans,$ where $U^L(\bs,t)$ is the truncated version of \eqref{eq:kl_U} at degree $L.$ Then, the truncated approximation to \eqref{eq:diffdyn} can be expressed compactly as
\begin{align}\label{eq:diffdyn_mat}
    {\bf \Psi\trans  D}_w{\bf U}^L_t = \boldsymbol{\xi}\circ{\bf \Psi\trans  D}_w{\bf U}^L_{t-1} + \boldsymbol{\eta}_t,
\end{align}
where $\boldsymbol{\xi}$ is a vector containing $\{e^{-(\xi_r+\xi_d l(l+1))}\}_{l=0,\ldots,L;m=-l,\ldots,l}$ and  $\boldsymbol{\eta} \sim \cN(0,{\bf D}_\eta),$ where ${\bf D}_\eta$ is a diagonal matrix consisting of $\frac{\sigma^2_Q(\kappa^2+l(l+1))^{-(\nu+1)}}{2(\xi_r+\xi_d l(l+1))}(1-e^{-2(\xi_r+\xi_d l(l+1))}).$
% The term $\hat{\cU}_{lm}(t)$ is the spherical harmonic transform of $\cU(t)$ evaluated at $(l,m)$ and corresponds to the $(l,m)^{\text{th}}$ entry of the vector ${\bf \Psi\trans  D}_w{\bf U}_t^L$. 

Equation \eqref{eq:diffdyn_mat} corresponds to a dynamic spatiotemporal model (DSTM; \citealp{cressie2011statistics}), with the full conditionals given by
\begin{align}\label{eq:postU}
    {\bf \Psi\trans D}_w{\bf U}^L_t \mid \cdot \sim \mathcal{N}(V_ta_t,V_t),
\end{align}
% \begin{align*}
%     a_0 &= \boldsymbol{\xi}\circ{\bf D}_\eta^{-1}{\bf \Psi\trans D}_w{\bf U}^L_1,   &V_0&= {\bf D}_\eta,\\
%     a_t &= \frac{({\bf \Psi\trans D}^2_w{\bf \Psi})^{-1}{\bf \Psi\trans D}_w({\bf Y}_t-\boldsymbol{\mu}_t) }{\sigma^2_\epsilon} + \boldsymbol{\xi}\circ{\bf D}_\eta^{-1}{\bf \Psi\trans D}_w({\bf U}^L_{t+1} + {\bf U}^L_{t-1}),  &V_t&= \left(\frac{({\bf \Psi\trans D}^2_w{\bf \Psi})^{-1}}{\sigma^2_\epsilon} + (\mathbf{1}+\boldsymbol{\xi}^2){\bf D}_\eta^{-1}\right)^{-1},\\
%     a_M &= \frac{({\bf \Psi\trans D}^2_w{\bf \Psi})^{-1}{\bf \Psi\trans D}_w({\bf Y}_M-\boldsymbol{\mu}_M) }{\sigma^2_\epsilon} + \boldsymbol{\xi}\circ{\bf D}_\eta^{-1}{\bf \Psi\trans D}_w{\bf U}^L_{M-1}, &V_M&= \left(\frac{({\bf \Psi\trans D}^2_w{\bf \Psi})^{-1}}{\sigma^2_\epsilon} + {\bf D}_\eta^{-1}\right)^{-1}. 
% \end{align*}
\begin{align*}
    a_0 &= \boldsymbol{\xi}\circ{\bf D}_\eta^{-1}{\bf \Psi\trans D}_w{\bf U}^L_1,\quad V_0 = {\bf D}_\eta,\\
    a_t &= \frac{({\bf \Psi\trans D}^2_w{\bf \Psi})^{-1}{\bf \Psi\trans D}_w({\bf Y}_t-\boldsymbol{\mu}_t(\boldsymbol{\tau})) }{\sigma^2_\epsilon} + \boldsymbol{\xi}\circ{\bf D}_\eta^{-1}{\bf \Psi\trans D}_w({\bf U}^L_{t+1} + {\bf U}^L_{t-1}),\\
    V_t&= \left(\frac{({\bf \Psi\trans D}^2_w{\bf \Psi})^{-1}}{\sigma^2_\epsilon} + (\mathbf{1}+\boldsymbol{\xi}^2){\bf D}_\eta^{-1}\right)^{-1},\\
    a_M &= \frac{({\bf \Psi\trans D}^2_w{\bf \Psi})^{-1}{\bf \Psi\trans D}_w({\bf Y}_M-\boldsymbol{\mu}_M(\boldsymbol{\tau})) }{\sigma^2_\epsilon} + \boldsymbol{\xi}\circ{\bf D}_\eta^{-1}{\bf \Psi\trans D}_w{\bf U}^L_{M-1}, \\
    V_M&= \left(\frac{({\bf \Psi\trans D}^2_w{\bf \Psi})^{-1}}{\sigma^2_\epsilon} + {\bf D}_\eta^{-1}\right)^{-1}. 
\end{align*}
The expansion \eqref{eq:kl_U} implies ${\bf U}^L_t = {\bf \Psi\Psi\trans D}_w{\bf U}^L_t$. Consequently, posterior sampling can be performed by first drawing the spectral coefficients ${\bf \Psi\trans D}_w{\bf U}^L_t$ and then recovering $\mathbf{U}^L_t$ via an inverse spherical harmonic transform.

The diffusion-based formulation in \eqref{eq:heat} induces a space–time non-separable structure for $U(\bs,t)$ whenever the diffusivity parameter $\xi_d > 0$, with separability attained if and only if $\xi_d = 0$. The degree of non-separability increases as $\xi_d$ departs from $0$ up and subsequently decreases as temporal dependence weakens for large $\xi_d$. A detailed derivation and characterization of this behavior are provided in Appendix~\ref{app:separability}.

Although assuming separability in space and time may potentially deteriorate statistical inference \citep{li2007nonparametric}, our numerical experiments suggest that assuming a space–time separable model for $U(\bs,t)$ has a negligible impact on changepoint estimation, even under parameter values that yield strongly non-separable dependence. Consequently, for simplicity, we adopt a space–time separable specification for $U(\bs,t)$ throughout the remainder of the paper. Under this assumption, the autoregressive vector $\boldsymbol{\xi}$ in \eqref{eq:postU} reduces to a scalar in $(0,1)$, and $\mathbf{D}_\eta$ simplifies to $\sigma^2_U \mathbf{D}_\alpha$ for some $\sigma^2_U > 0$.

\section{Simulation Study}\label{sec:simulation}
We conduct simulation studies to evaluate the ability of the multinomial probit model to improve changepoint estimation by leveraging spatial dependence. We then examine the impact of spectral truncation on estimation accuracy and computational efficiency.

\subsection{Data generation and evaluation metric}\label{sec:simdat}
To ensure that the numerical experiments reflect realistic conditions, we generated data based on a reanalysis dataset consisting of 60 months of monthly global aerosol optical depth (AOD) estimates spanning January 1985 to December 1989. We deliberately selected this time window because it contains no known major atmospheric events that could have introduced changepoints in AOD, allowing synthetic changepoints to be introduced at known locations and times for controlled assessment. Using such physically informed data, rather than observations generated from a purely statistical model, allows for a more realistic evaluation of the methodology.

We obtained the AOD data from the Modern-Era Retrospective Analysis for Research and Applications, Version 2 (MERRA-2) reanalysis dataset, using the “TOTEXTTAU” variable from the Global Modeling and Assimilation Office \citep{merra2_aod}. This variable measures extinction optical thickness at 550 nm and has a spatial resolution of $1^\circ \times 1^\circ$, corresponding to a $180 \times 360$ global grid and $N = 64{,}800$ total spatial locations. To remove systematic temporal structure, we first eliminate seasonality at each location using seasonal–trend decomposition based on LOESS (STL). The data are then log-transformed and standardized so that the baseline mean function 
$\mu_1(\mathbf{s},t)$ is approximately constant across space and time.

Synthetic mean shifts are introduced by adding a constant to all post-changepoint observations, with effect sizes in $\{1,1.5,2\}$ corresponding to low, moderate, and high signal-to-noise ratios. This yields mean functions $\mu_1(\mathbf{s}, t) = \mu_1$ and $\mu_2(\mathbf{s}, t) \in {\mu_1 + 1, \mu_1 + 1.5, \mu_1 + 2}$ in model \eqref{eq:model}. 

The changepoints $\tau(\bs)$ are generated using two different approaches. In both cases, we begin by simulating an auxilliary variable $\tilde{\tau}(\mathbf{s})$ from a mean-zero multivariate normal distribution with Mat\'ern covariance with smoothness $\nu = 1$ and inverse range parameter $\kappa \in {3,5,100}$, representing high, moderate, and near-zero spatial correlation, respectively. In the first approach, $\tilde{\tau}$ is transformed via min-max scaling:
\begin{equation}\label{eq:tau_method1}
\tau_1(\bs) := \left\lfloor \frac{\tilde{\tau}(\bs) - \min(\tilde{\tau}(\bs))}{\max(\tilde{\tau}(\bs))-\min(\tilde{\tau}(\bs))}*49 + 6 \right\rfloor.
\end{equation}
This transformation maps $\tilde{\tau}(\bs)$ to the discrete set $\{6,7,\ldots,55\}$ and induces marginal probabilities $\pi_k(\bs)$ that favor changepoints near the center of the time series.
% Under this transformation,  the marginal probabilities are given by $\pi_k(\bs) = \mathbb{P}\left( \frac{k-6}{49} < \frac{\tilde{\tau}(\bs) - \min(\tilde{\tau}(\bs))}{\max(\tilde{\tau}(\bs))-\min(\tilde{\tau}(\bs))} \leq \frac{k-5}{49}\right),$ which assigns greater weight to the changepoints near the middle of the time series. 

In the second approach, changepoints are generated via the Gaussian CDF $\Phi(\cdot)$:
\begin{equation}\label{eq:tau_method2}
\tau_2(\bs) := \left\lfloor \Phi(\tilde{\tau}(\bs)) \times 50 + 6 \right\rfloor .
\end{equation}
This transformation yields an approximately discrete uniform marginal probabilities over the same support ($\pi_k(\bs) \approx 1/50$ for all $k\in\{6,7,\ldots,55\}$), which is commonly adopted as a non-informative prior in the absence of prior knowledge about the changepoint process. 
Notably, in both approaches, the changepoints are \textit{not} generated from the proposed multinomial probit model \eqref{eq:cpmodel2}-\eqref{eq:cpmodel1}. Instead, we intentionally introduce model misspecification to assess the robustness of our method to deviations from the assumed changepoint prior. For each parameter setting, we generate 100 independent simulation replicates.

We evaluate the changepoint model performance through estimation accuracy rather than detection accuracy. For spherical data, aggregating binary detection errors can be difficult to evaluate, as it corresponds to misclassified surface area rather than conventional false positive or false negative rates. Estimation error, by contrast, provides a natural and spatially coherent summary of performance and implicitly reflects detection accuracy, since incorrectly classified locations contribute directly to the overall error. Let $\hat{\boldsymbol{\tau}}$ denote the posterior mean of the changepoint process. To quantify estimation accuracy on the sphere, we take $\hat{\boldsymbol{\tau}}$ as the estimated changepoint and consider its $L^2$ distance from the true changepoint process, which can be viewed as a spherical analogue of root mean squared error (RMSE). Given that the surface area of a unit sphere is $4\pi$, we define the generalized RMSE as 
\begin{align}\label{eq:rmse}
   \text{g-RMSE}(\hat{\boldsymbol{\tau}}) = \frac{1}{4\pi}\lVert \boldsymbol{\tau}-\hat{\boldsymbol{\tau}} \rVert_{L^2(\mathbb{S}^2)} = \sqrt{\frac{1}{4\pi}\sum_{i=1}^N w(\bs_i)(\tau(\bs_i)-\hat{\tau}(\bs_i))^2}.
\end{align}

\subsection{Effect of spatial correlation on changepoint estimation}\label{sec:sim1}
To assess the benefit of explicitly modeling spatial dependence in changepoints, we compare the proposed multinomial probit model (MPM) in \eqref{eq:cpmodel2}-\eqref{eq:cpmodel1} with a baseline model that assumes no spatial dependence in the changepoints, hereafter denoted IND. Under IND, the prior changepoint probabilities must be spatially constant and sum to one over $k=1,\ldots,M$. Accordingly, we adopt a non-informative discrete uniform prior $\pi_k(\bs)=1/M$, which yields posterior estimates that coincide with the maximum likelihood estimator.

Figure~\ref{fig:sim1_rmse}(a) summarizes estimation accuracy under varying mean shift magnitudes and changepoint correlation strengths for the changepoints $\boldsymbol{\tau}_1$ generated via \eqref{eq:tau_method1}. Across all settings, MPM consistently achieves lower g-RMSE than IND. This performance gap increases with stronger spatial correlation, demonstrating that MPM effectively captures and exploits spatial dependence to improve changepoint estimation. Moreover, the relative advantage of MPM becomes more pronounced as the mean-shift signal weakens, indicating that borrowing strength across space is particularly beneficial in low signal-to-noise scenarios.

\begin{figure}[!ht]
\centering
\includegraphics[width=0.9\textwidth]{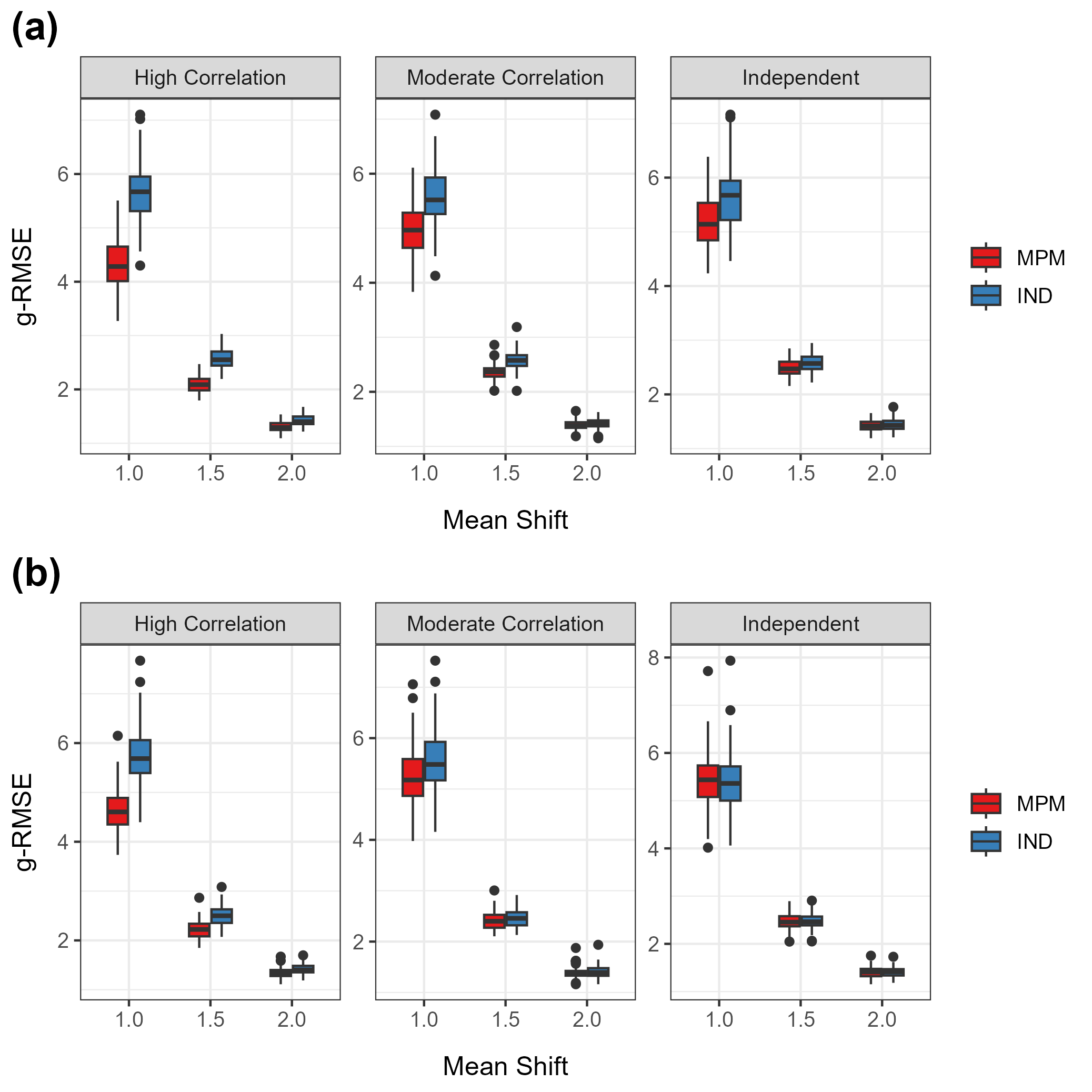}
\caption{\label{fig:sim1_rmse} Boxplots of g-RMSE produced by MPM and IND under different mean shift values (1, 1.5 and 2.0) and changepoint correlation strength for (a) changepoints $\boldsymbol{\tau}_1$ generated by \eqref{eq:tau_method1} and (b) changepoints $\boldsymbol{\tau}_2$ generated by \eqref{eq:tau_method2}. The x-axis displays the mean shift signal, and the y-axis displays g-RMSE. } 
\end{figure}

Surprisingly, MPM also outperforms IND even when spatial correlation is negligible ($\kappa = 100$), as observed in the “Independent” case in Figure~\ref{fig:sim1_rmse}(a). This behavior reflects differences not only in how the two models encode spatial dependence, but also in how they treat the marginal distribution of the changepoint. Specifically, IND fixes $\pi_k(\bs) = 1/M$ for all $k$, which is misspecified for the data-generating mechanism underlying $\boldsymbol{\tau}_1$. In contrast, MPM estimates the marginal changepoint probabilities from the data through the latent process $\mu_Z(\bs)$, providing additional flexibility even in the absence of spatial correlation. This interpretation is corroborated by results based on changepoints $\boldsymbol{\tau}_2$ generated via \eqref{eq:tau_method2}, which yields approximately uniform marginal distribution for $\pi_k(\bs)$ by construction. Indeed, under the “Independent” setting in Figure~\ref{fig:sim1_rmse}(b), the performance of MPM and IND is nearly identical, with IND slightly favored due to its simpler structure and exact alignment with the data-generating process. When spatial correlation is present, MPM again yields lower g-RMSE, although the improvement is more modest than in the $\boldsymbol{\tau}_1$ setting. Taken together, these results demonstrate that MPM is robust to both correlation and marginal distribution misspecification. While IND performs well when its assumptions are exactly satisfied, MPM maintains competitive performance in that regime and offers substantial gains when either spatial dependence or marginal structure deviates from uniformity. This robustness is particularly valuable in applications where prior information about the changepoint distribution is limited or unreliable.

\subsection{Effect of truncation error on changepoint estimation}\label{sec:sim2}
We now examine how truncating the spectral representations affects changepoint estimation. In the proposed framework, two latent processes are approximated by finite expansions at degree $L$: the spatial field $\mu_Z$ by $\mu_Z^L$ in \eqref{eq:kl_exp}, and the spatiotemporal process $U$ by $U^L$ in \eqref{eq:kl_U}. These truncations introduce distinct sources of approximation error.

Since $\mu_Z$ enters the changepoint model \eqref{eq:cpmodel2}-\eqref{eq:cpmodel1} directly through the multinomial probit construction, the impact of truncation in $\mu_Z$ can be analyzed theoretically using probabilistic bounds, as established in Section~\ref{sec:sh} and Appendix~\ref{app:appendix_trunc}. In contrast, truncation of the spatiotemporal noise process $U$ affects changepoint inference only indirectly through the likelihood. As a result, its effect is data-dependent and there is no global, model-level characterization analogous to that available for $\mu_Z$. We therefore rely on numerical experiments to assess the impact of truncating $U(\bs,t)$.

By the spherical harmonics addition theorem, the marginal spatial covariance of $U(\bs,t)$ is given by
\begin{align}\label{eq:truecov}
    C_s({\bf s},{\bf s'}; \kappa,\nu,\sigma^2_U) &= \sigma^2_U\sum_{l=0}^\infty \sum_{m=-l}^l (\kappa^2 + l(l+1))^{-(\nu+1)} \psi_{lm}({\bf s})\psi_{lm}({\bf s'}) \nonumber \\
    & = \frac{\sigma^2_U}{4\pi}\sum_{l=0}^\infty (2l + 1)(\kappa^2 + l(l+1))^{-(\nu+1)} P_l({\bf s} \cdot {\bf s'}).
\end{align}
Although this series does not have a closed-form expression, it is convergent \citep{lang2015isotropic} and can be evaluated numerically. This induces the innovation covariance matrix $\bSigma_U$ in the temporal innovation process  $\mathbf{U}_t - \xi \mathbf{U}_{t-1} \sim \mathcal{N}(0,\bSigma_U)$ and thus determines the prior distribution of $U(\bs,t).$

We define the truncated marginal spatial covariance as
\begin{align}\label{eq:trunccov}
  C^L_s(\bs,\bs';\kappa^L,\nu^L,\sigma^2_{U^L}) = \frac{\sigma^2_{U^L}}{4\pi}\sum_{l=0}^L (2l + 1)({\kappa^L}^2 + l(l+1))^{-(\nu^L+1)} P_l({\bf s} \cdot {\bf s'}),   
\end{align}
where the parameters $\kappa^L,\nu^L,\sigma^2_{U^L}$ are distinguished from $\kappa,\nu,\sigma^2_U$ to emphasize that they characterize the truncated process $U^L$ as opposed to the untruncated process $U$.

Let $\hat{\boldsymbol{\tau}}$ denote the posterior mean of the changepoint when $\bSigma_U$ is computed using the full covariance $C_s$, and $\hat{\boldsymbol{\tau}}^L$ the posterior mean using the spectral approach as outlined in Section \ref{sec:st}, which is mathematically equivalent to computing $\bSigma_U$ using $C_s^L.$ To ensure that any difference between $\hat{\boldsymbol{\tau}}$ and $\hat{\boldsymbol{\tau}}^L$ arises solely from truncation of the spatiotemporal process $U$, we use the simulation setting with changepoints $\boldsymbol{\tau}_2$ and $\kappa=100$ from Section~\ref{sec:simdat} and fit the data using IND model from Section~\ref{sec:sim1} such that the changepoint probabilities $\mathbb{P}(\tau^L(\bs) = k)$ is constant across $L$. To quantify the loss in changepoint estimation accuracy due to truncation of $U$, we examine the difference between g-RMSE$(\hat{\boldsymbol{\tau}}^L)$ and g-RMSE$(\hat{\boldsymbol{\tau}})$. Direct computation of $\hat{\boldsymbol{\tau}},$ however, requires storing the full $64800 \times 64800 $ covariance matrix and is infeasible due to memory and computational constraints. Instead, we approximate $\hat{\boldsymbol{\tau}}$ by studying the behavior of $\hat{\boldsymbol{\tau}}^L$ for increasing values of $L$ and extrapolating to the limit $L \to \infty$. 

To ensure $\hat{\boldsymbol{\tau}}^L \to \hat{\boldsymbol{\tau}}$ as $L \to \infty,$ we construct coupled MCMC chains across truncation levels by first initializing the chains at the same parameter values for every $L$. We then generate paired samples across $L$ using a common source of randomness in the following way: Let $F_{\vartheta,L}(\cdot;\varphi)$ denote the CDF of the full conditional for parameter $\vartheta$ under truncation level $L$, conditional on the remaining parameters $\varphi$. In principle, the $i^\text{th}$ draw from the full conditional can be written as 
\[\vartheta^L_i = F_{\vartheta,L}^{-1}(u_i;\varphi^L_{i-1}),\]
where the same $u_i \sim \text{Unif}(0,1)$ is used across $L.$ This construction couples the Markov chains across all truncation levels, allowing us to isolate the truncation error from Monte Carlo variability.

Let $F_{\vartheta}(\cdot;\varphi)$ be the full conditional CDF under the untruncated model, and define $\vartheta_i := F^{-1}_{\vartheta}(u_i;\varphi_{i-1}).$ For all parameters $\vartheta \neq \mathbf{U}$, the functional form of the full conditional distribution is identical under truncation, and any difference between $F_{\vartheta}$ and $F_{\vartheta,L}$ arises solely through conditioning on $U^L$ versus $U$. For $U$, the full conditional distribution differs in two distinct ways: (i) a structural difference induced by approximating the spatial covariance kernel $C_s$ with its truncated version $C_s^L$, and (ii) replacement of the conditioning parameters $(\xi, \sigma_\epsilon^2, \mu_1, \mu_2, \boldsymbol{\tau}, \kappa, \nu, \sigma_U^2)$ by their truncated counterparts.
The structural difference in (i) can be controlled via the decomposition
\begin{align*}
    \sup_{\bs,\bs'} |C_s(\bs,\bs';\varphi) - C^L_s(\bs,\bs';\varphi^L)|  &\leq   \sup_{\bs,\bs'} |C_s(\bs,\bs';\varphi) - C^L_s(\bs,\bs';\varphi)| \\
   &\hspace{2mm} +  \sup_{\bs,\bs'} |C^L_s(\bs,\bs';\varphi) - C^L_s(\bs,\bs';\varphi^L)|. 
\end{align*}
The first term converges to zero as $L \to \infty$ by convergence of the truncated spherical harmonic expansion, and the second term vanishes provided the associated conditioning parameters converge. Consequently, the convergence of $(\xi^L, \sigma^L_\epsilon, \mu_1^L, \mu_2^L, \boldsymbol{\tau}^L, \kappa^L, \nu^L, \sigma_{U^L})$ to their untruncated counterparts in the iteration $i-1$ guarantees the convergence of posterior samples $\mathbf{U}^L_i$ to $\mathbf{U}_i$. For $i = 1$, this holds trivially since all parameters are initialized identically across truncation levels. Given $\mathbf{U}_1^L \to \mathbf{U}_1$, the continuity of the corresponding full conditional CDFs implies $\vartheta_1^L \to \vartheta_1$ for all $\vartheta\neq U$. The same argument applies inductively to subsequent iterations, establishing convergence of the coupled posterior samples across truncation levels.

In practice, we realize the coupling by seeding the random number generator with a fixed integer before drawing from each pair of full conditionals. Because modern pseudorandom generators are deterministic maps of the seed, this produces the same underlying random variables for both the full and truncated conditionals, thereby achieving the same coupling as the theoretical construction above.

Following the outlined sampling procedure, $\boldsymbol{\tau}^L$ converges to $\boldsymbol{\tau}$ as $L\to\infty$ at each MCMC iteration, implying convergence of $\text{g-RMSE}(\hat{\boldsymbol{\tau}}^L)$ to $\text{g-RMSE}(\hat{\boldsymbol{\tau}})$.  Since the truncation error decreases with larger $L$, we expect g-RMSE$(\hat{\boldsymbol{\tau}}^L)$ to decrease as $L$ increases and eventually stabilize at g-RMSE$(\hat{\boldsymbol{\tau}}).$ We model this behavior using an exponential decay function \citep[e.g.,][]{aston2012radioactive, liu2021performance},
\begin{equation}\label{eq:exp}
\text{g-RMSE}(\hat{\boldsymbol{\tau}}^L) = a \exp(-bL) + \text{g-RMSE}(\hat{\boldsymbol{\tau}}).
\end{equation}
Given the maximum admissible truncation degree $L_{\max}=K/2-1=89$, we consider $L\in\{9,19,49,89\}$. Figure~\ref{fig:sim2} displays the fitted $\text{g-RMSE}(\hat{\boldsymbol{\tau}}^L)$ as a function of $L$ for different mean shift magnitudes. As expected, both the initial error level $a$ and the asymptotic value $\text{g-RMSE}(\hat{\boldsymbol{\tau}})$ increase as the signal weakens. The estimated decay rates are $b=0.1088, 0.1203,$ and $0.1245$ for mean shifts $1, 1.5,$ and $2$, respectively, indicating that truncation has a larger impact when the signal-to-noise ratio is low. Nevertheless, for all settings, both the fitted and observed values of $\text{g-RMSE}(\hat{\boldsymbol{\tau}}^{89})$ lie very close to the estimated asymptotic error, demonstrating that truncation-induced loss in changepoint estimation accuracy is negligible at the grid resolution considered here ($K=180$).

\begin{figure}[!ht]
\centering
\includegraphics[width=0.9\textwidth]{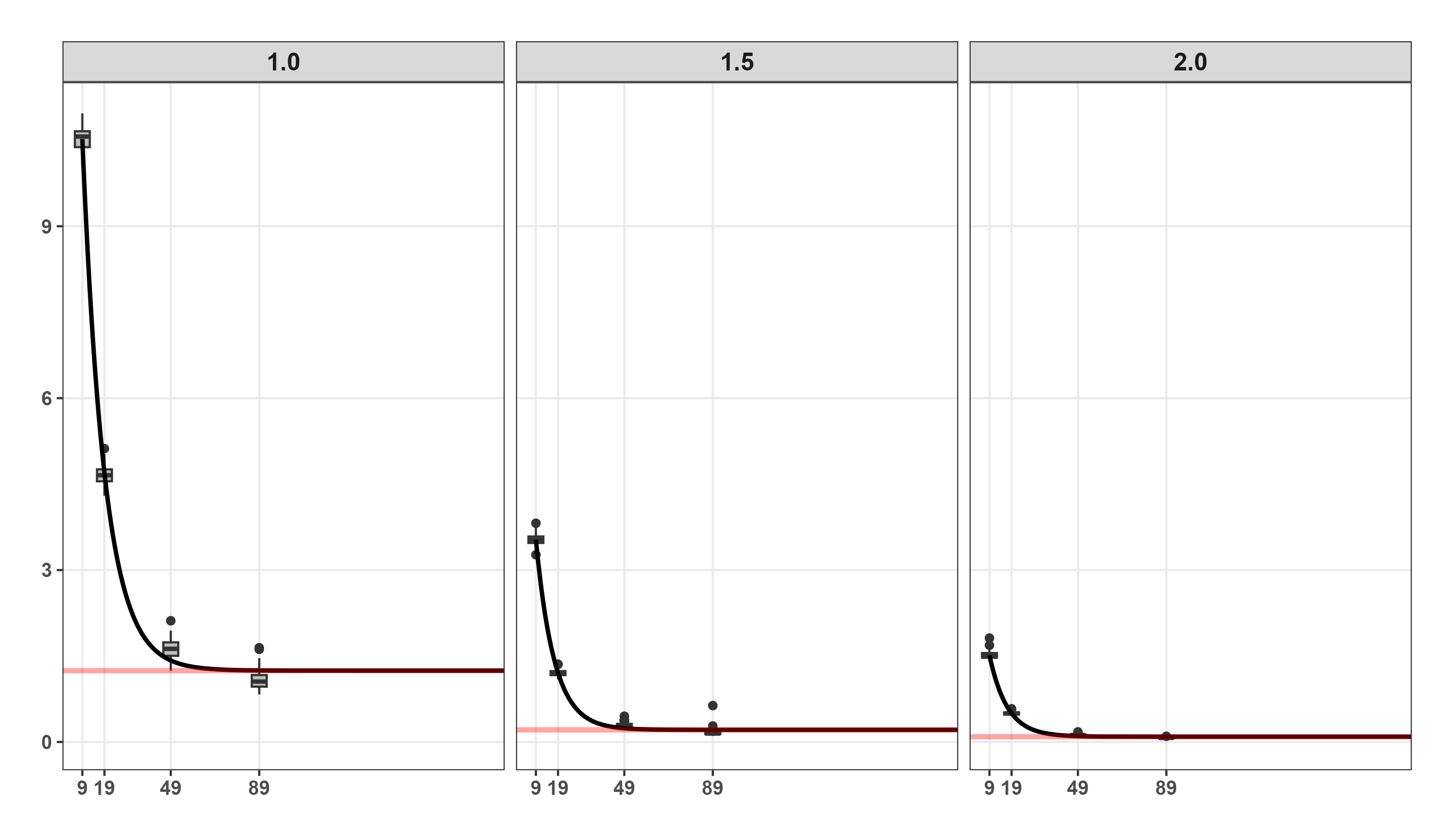}
\caption{\label{fig:sim2} Fitted g-RMSE$(\hat{\boldsymbol{\tau}}^L)$ as a function of $L$ using \eqref{eq:exp} for different mean shift values $\{1, 1.5, 2\}$. The boxplots represent the observed g-RMSE$(\hat{\boldsymbol{\tau}}^L)$ across 100 simulations. The red line represents estimated g-RMSE$(\hat{\boldsymbol{\tau}}).$ } 
\end{figure}

\subsection{Computational efficiency}
We now assess the computational gains afforded by two key components of the proposed methodology:
(i) replacing the full spatial covariance matrix $\bSigma_U$
 with a truncated spherical harmonic representation, and
(ii) replacing Metropolis-Hastings sampling for the changepoints with a Gibbs sampler enabled by the multinomial probit formulation (Section~\ref{sec:cpmodel}).

To provide a controlled comparison, we first describe the implementation of baseline method that based on full covariance matrix $\bSigma_U$ and Metropolis-Hastings (MH) sampling. Since computing $\bSigma_U$ directly using \eqref{eq:truecov} for large $N$ is infeasible for most computing systems due to memory and runtime constraints, we restrict our comparison to settings with modest $N$. It is important to note, however, that the relative computational advantage of the spherical harmonics approach is expected to become even more pronounced for large $N$ as the cost of evaluating and factorizing $N$ grows superlinearly. For Metropolis-Hastings sampling, doing a single ``block" update of the vector $\boldsymbol{\tau}$ is impractical due to difficulty in designing a multivariate proposal with a reasonable acceptance rate, causing the chain to move extremely slowly through the state space and leading to poor mixing. Therefore, we adopt a component-wise MH scheme, sequentially updating $\tau(\bs_i)$ for $i = 1,\ldots,N$. 

\begin{table}[h!]
\centering
\begin{tabular}{lccccc}
\hline\hline
$N$ & $L$ & \# of basis & SH \eqref{eq:trunccov} + MPM & $\bSigma_U$ \eqref{eq:truecov} + MPM & $\bSigma_U$ \eqref{eq:truecov} + MH \\\hline
$800$ & $9$ & 100 & 0.036 s & 11.824 s & 674.702 s\\
$3,200$ & $19$ & 400 & 0.119 s & 149.442 s & 14721.650 s\\
$20,000$ & $49$ & 2,500 & 0.809 s & -- & -- \\
$64,800$ & $89$ & 8,100 & 2.946 s & -- & -- \\
\hline
\end{tabular}
\caption{\label{tab:computation} Average computation speed (time/iteration) over 100 iterations for (1) spherical harmonics (SH) representation with multinomial probit model (MPM) for changepoints, (2) full covariance matrix with MPM, and (3) full covariance matrix with MH update.}
\end{table}

Table \ref{tab:computation}  reports the average runtime per iteration for three approaches: (1) spherical harmonics (SH) representation with multinomial probit model (MPM) and Gibbs sampling, (2) full covariance $\bSigma$ with MPM, and
(3) full covariance $\bSigma_U$ with component-wise MH sampling for changepoints. The results demonstrate that replacing the full covariance with a spherical harmonic representation yields dramatic reductions in computation time with minimal loss in changepoint estimation accuracy for sufficiently large truncation level $L$. (Sections~\ref{sec:sh} and~\ref{sec:sim2}). While the full covariance approach provides a direct representation of spatial dependence, it becomes computationally prohibitive even at moderate $N$, both in memory usage and runtime. In contrast, the spectral approach enables practical analysis of large-scale spatial data sets that would otherwise be infeasible. 

Beyond the spectral representation, the Gibbs sampling scheme enabled by the multinomial probit model provides an additional and substantial efficiency gain. Component-wise MH requires sequential updates with acceptance–rejection steps at every spatial location, whereas the Gibbs sampler updates all changepoints jointly in a single step. When combined with the spherical harmonic representation, this yields per-iteration runtimes on the order of seconds even for large $N$, compared to minutes or hours for MH-based alternatives. Together, these results show that the two methodological contributions -- fast spherical harmonic transform via spectral representation and Gibbs sampling via MPM -- act synergistically to enable efficient inference for large-scale spatial changepoint problems.

\section{Data Application}\label{sec:data}
We apply our method to 60 months of global stratospheric AOD data spanning January 1989 to December 1993, obtained from the same MERRA-2 reanalysis source described in Section~\ref{sec:simdat}. This time window was selected because the AOD record contains well-documented structural changes beginning in June 1991 following the eruption of Mount Pinatubo, the largest volcanic eruption in recent history. The eruption injected nearly 20 megatons of sulfur dioxide into the stratosphere, producing a global aerosol cloud that encircled the Earth within weeks and resulted in the most substantial perturbation to the stratospheric aerosol layer since the eruption of Krakatau in 1883 \citep{self1993}.

Prior to changepoint analysis, we preprocess the data to remove seasonality at each location using STL decomposition and apply a logarithmic transformation. We then estimate a linear temporal trend using observations prior to June 1991, the month of the eruption. When the estimated trend is statistically significant at the 0.05 level, it is removed from both the pre- and post-eruption periods to ensure that the pre-changepoint mean is approximately constant across time.

Previous work by \cite{shi2025tracing} analyzed spatial changepoints for the same dataset using a downsampled $16 \times 48$ grid, necessitated by computational constraints. Their analysis identified changepoints occurring between June and September 1991, with the earliest changes concentrated between approximately $3.5^\circ$S and $34^\circ$N latitude. Following their work, we model the pre- and post-changepoint mean processes as
\begin{align*}
    \mu_1(\bs,t) &= \beta_0\\
    \mu_2(\bs,t) &= \beta_1(\bs) + \beta_2(\bs)(t-\tau(\bs)),
\end{align*}
where $\beta_1(\bs)$ and $\beta_2(\bs)$ capture spatially varying changes in level and trend following the changepoint. For computational simplicity, we assume $\boldsymbol{\beta}_i \sim \mathcal{N}(\beta_i^{F},\sigma^2_{\beta_i}\mathbf{I})$ for $i=1,2$. For identifiability, the variance parameters $\sigma^2_{\beta_i}$ are treated as ridge penalties and selected using the Watanabe–Akaike Information Criterion (WAIC) \citep{shi2025tracing}.

At each spatial location, the changepoint is estimated using the posterior mode of $\tau(\bs)$, with $\tau(\bs)=M$ interpreted as no changepoint. Figure~\ref{fig:aod}(a) displays the resulting heatmap of detected changepoints, with locations for which no changepoint is detected shown in white. Our method detects changepoints at nearly all locations, with dates ranging from May through October 1991. Under our model specification \eqref{eq:model}, a changepoint at May 1991 corresponds to a mean shift being observed in June 1991, which is consistent with the eruption date. The estimated changepoints show a pattern driven more by latitude than longitude, with the two earliest changepoints May 1991 and June 1991 concentrated between latitudes $35^\circ N$ through $14^\circ S$. This pattern is consistent with existing literature, which reports that the Pinatubo aerosol layer circled the Earth in 21 days and had spread to latitudes around $30^\circ N$ and $10^\circ S$ in the same period \citep{self1993}. The results are also consistent with the findings of \cite{shi2025tracing}, while providing finer spatial detail due to the higher resolution of the present analysis. Compared to earlier work, our high-resolution model reveals localized heterogeneity that is not visible on coarser grids. In particular, a small subset of locations exhibits no detectable changepoint -- a feature that was previously obscured by the coarser grid used in \cite{shi2025tracing}. Figure~\ref{fig:aod}(b) shows representative time series from these locations. Although these series display modest increases around the eruption period, the magnitude of change is small relative to the background variability, making it reasonable not to classify these locations as having a significant mean shift. Overall, the results are consistent with established scientific understanding of the Pinatubo eruption, but our high-resolution changepoint analysis can uncover subtle, localized deviations that would be missed at coarser scales.    

\begin{figure}[!ht]
\centering
\includegraphics[width=0.95\textwidth]{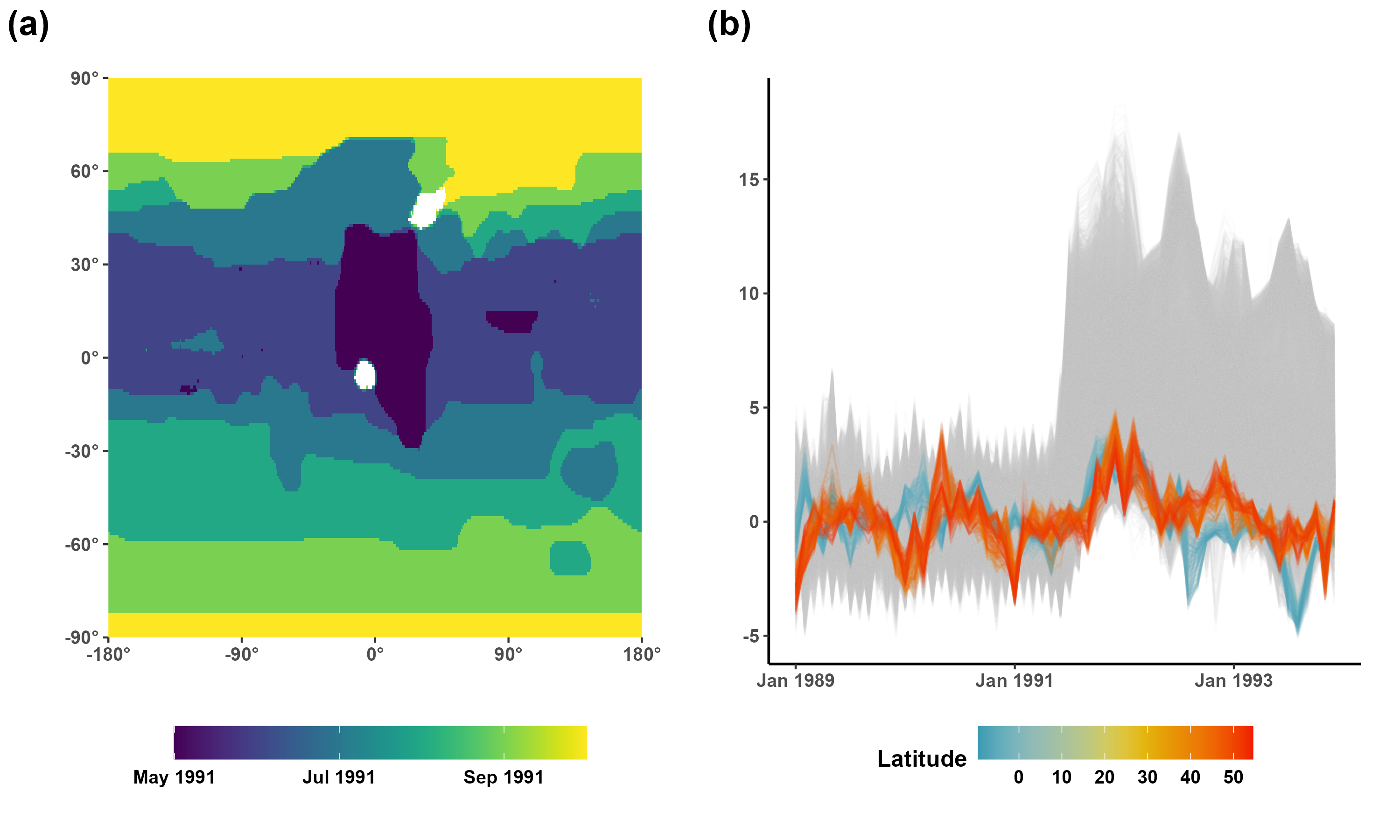}
\caption{\label{fig:aod} (a) Heatmap of detected changepoints in stratospheric AOD data. White color indicates that no changepoints were detected. (b) Time series of locations where no changepoints were detected. Grey lines represent time series of locations with changepoints detected.} 
\end{figure}

\section{Discussion}\label{sec:disc}
We propose a scalable Bayesian framework for spatial changepoint detection on the sphere, designed to address the computational challenges posed by high-resolution global data. The key methodological contribution is the replacement of Metropolis–Hastings updates for spatially varying changepoints with a fully Gibbs-sampled procedure enabled by a multinomial probit formulation. This modeling choice preserves spatial dependence while substantially improving computational efficiency and stability in high dimensions. In parallel, we develop a spectral inference framework based on fast spherical harmonic transformations that enables efficient Bayesian modeling of large spatial and spatiotemporal processes without explicit construction of dense covariance matrices.

The proposed methodology is robust to varying strengths and forms of spatial correlation and achieves substantial computational gains while incurring only minimal loss in estimation accuracy.  The application to global stratospheric aerosol optical depth data demonstrates that the proposed methodology can recover scientifically meaningful changepoint patterns at native spatial resolution, revealing localized structure that is obscured by spatial aggregation. While the methodological development is motivated by changepoint detection, the underlying concepts and computational strategies are broadly applicable to large-scale spherical data analysis.

Several directions for future research are worth exploring. One promising avenue is the incorporation of modern generative models as flexible priors for spatial and spatiotemporal latent processes. Recent developments in score-based diffusion models and variational autoencoder–type constructions provide principled mechanisms for approximating high-dimensional Gaussian and non-Gaussian distributions while retaining compatibility with Bayesian inference. In particular, generative priors trained to approximate Gaussian processes offer a potential route to scalable surrogates for complex covariance structures that are otherwise computationally prohibitive \citep{semenova2022priorvae, cardoso2025predictive}. Integrating such priors into hierarchical changepoint models could enable richer, data-adaptive representations of spatial dependence while preserving efficient posterior sampling. Establishing theoretical guarantees for these approximations and understanding their implications for posterior uncertainty quantification remain important open problems.

\section*{Acknowledgments}
We thank Lyndsay Shand for assistance with data retrieval and for sharing domain expertise relevant to the atmospheric application considered in this paper.

% \section*{Data Availability Statement}
% The data used in this study are publicly available from NASA’s Modern-Era Retrospective Analysis for Research and Applications, Version 2 (MERRA-2). The aerosol optical depth variable (TOTEXTTAU) can be accessed through the Global Modeling and Assimilation Office \citep{merra2_aod}. Processed data and analysis code supporting the findings of this study are available from the corresponding author upon reasonable request.

\bibliographystyle{apalike} % Style BST file
\bibliography{sample.bib}       % Bibliography file (usually '*.bib')

\clearpage
\phantomsection 
%\addcontentsline{toc}{section}{Appendix} % Add an entry to the ToC
\section*{Appendix}
\pdfbookmark[0]{Appendix}{app_heading} % Add a corresponding PDF bookmark
\appendix % Changes numbering scheme from chapters/sections to A, B, C
% \phantomsection
% \pdfbookmark[0]{Appendix}{appendix}
% \addcontentsline{toc}{section}{Appendix}

\section{SPDE representation and spectral properties of Mat\'ern Fields on the sphere}\label{app:appendix_a}

Let $\mathcal{D}$ be a compact manifold and $\phi \in L^2(\mathcal{D}).$ The eigenfunctions $\{E_k\}_{k=1}^\infty$ of negative Laplacian $-\Delta E_k = \lambda_k E_k$ forms an orthonormal basis for $L^2(\mathcal{D})$, and $\phi$ can be expressed as 
\begin{align*}
    \phi &= \sum_{k=1}^\infty \hat{\phi}(k)E_k,
\end{align*}
where $\hat{\phi}(k) = \langle \phi,E_k \rangle_{L^2(\mathcal{D})}.$
\begin{defn} 
\label{defn:fractional}
Let $\mathcal{L}$ be a non-negative, self-adjoint operator on a Hilbert space $\mathcal{H}$ with eigenvalues $\{\lambda_k\}_{k=0}^\infty$ and eigenfunctions $\{E_k\}_{k=0}^\infty$, where $ \mathcal{L} E_k = \lambda_k E_k$ and $\{E_k\}_{k=0}^\infty$ forms an orthonormal basis of $\mathcal{H}$. Then, for any $\phi \in \mathcal{H}$, the fractional power $\mathcal{L}^\alpha$ for $\alpha > 0$ is defined by:
\[
\mathcal{L}^\alpha \phi = \sum_{k=1}^\infty \lambda_k^\alpha \langle \phi, E_k \rangle_\mathcal{H} E_k,
\]
provided the series converges in $\mathcal{H}$.
\end{defn}

\begin{lem}
\label{lem:spec_identity}
Let $\lambda_k, k = 1,2,\ldots$ be the eigenvalues of $-\Delta$ on $\mathcal{D}.$ Let $\phi \in L^2(\mathcal{D}).$ Then,
\[\reallywidehat{(\kappa^2-\Delta)^{\alpha/2}\phi}(k) = (\kappa^2 + \lambda_k)^{\alpha/2}\hat{\phi}(k).\]
\end{lem}
\begin{proof}
Let $E_k, k=1,2,\ldots$ be the eigenfunctions of $-\Delta$ on $L^2({D})$ with eigenvalues $\lambda_k.$ Since
\begin{align*}
-\Delta\phi &= \sum_{k=1}^\infty \hat{\phi}(k)(-\Delta E_k) = \sum_{k=1}^\infty \lambda_k\hat{\phi}(k)E_k\\
\intertext{and}
\kappa^2\phi(s) &= \kappa^2\sum_{k=1}^\infty \hat{\phi}(k)E_k(s),
\end{align*}
the operator $(\kappa^2 - \Delta)$ is non-negative definite and self-adjoint with spectral representation given by
\[(\kappa^2-\Delta)\phi(s) = \sum_{k=1}^\infty (\kappa^2 + \lambda_k)\hat{\phi}(k)E_k(s).\]
The result follows from Definition \ref{defn:fractional}.
\end{proof}

\begin{defn} 
\label{defn:gwn}
A random field $\mathcal{W}$ is said to be a Gaussian white noise on $\mathcal{D}$ if, for any set of finite test functions $\{\phi_k \in L^2(\mathcal{D})\}_{k=1}^n,$ $(\langle \phi_1,\mathcal{W}\rangle_{L^2(\mathcal{D})}, \ldots, \langle \phi_n,\mathcal{W}\rangle_{L^2(\mathcal{D})}$ is multivariate Gaussian with
\begin{align*}
    \mathbb{E}[\langle \phi_k,\mathcal{W}\rangle_{L^2(\mathcal{D})}] &= 0\\
    \text{Cov}(\langle \phi_i,\mathcal{W}\rangle_{L^2(\mathcal{D})},\langle \phi_j,\mathcal{W}\rangle_{L^2(\mathcal{D})}) &= \langle \phi_i,\phi_j\rangle_{L^2(\mathcal{D})}. 
\end{align*}
\end{defn}

\begin{thm}
    The solution to \eqref{eq:SPDE} on $\mathbb{S}^2$ has a spectral representation given by \eqref{eq:kl_exp}, where the spectral density is given by
    \[S_l = \kappa^2+l(l+1))^{-(\nu+1)} \]
\end{thm}
\begin{proof}
From Definition \ref{defn:gwn}, we have $(\hat{\mathcal{W}}_1,\ldots,\hat{\mathcal{W}}_n) \overset{d}{=} \mathcal{N}(0,\mathbf{I}_n)$ for all $n \in \mathbb{N}.$ By Lemma~\ref{lem:spec_identity}, applying the spectral representation on both sides of $\eqref{eq:SPDE}$ gives 
\[(\kappa^2 + \lambda_k)^{(\nu+d/2)/2}\hat{X}(k) = \xi_k,\]
where $\xi_k \overset{\mathrm{iid}}{\sim} \mathcal{N}(0,1).$ Thus, 
\[\hat{X}_k \overset{\mathrm{ind}}{\sim} \mathcal{N}(0,(\kappa^2 + \lambda_k)^{-(\nu+d/2)}).\]
The result follows from noting that the spherical harmonics function $\psi_{lm}$ is the eigenfunction of negative Laplacian on $\mathbb{S}^2$ with eigenvalue $l(l+1).$

\end{proof}

\section{Probabilistic guarantees for truncated changepoint process}\label{app:appendix_trunc}
\begin{prop}\label{prop:conv}
    $\mu_Z^L(\bs) \to \mu_Z(\bs)$ as $L \to \infty$ almost surely (in $\mathbb{P})$ and pointwise (in $\bs$) for $\nu > 1/2,$ with 
    \[\mathbb{P}(\lvert \mu_Z^L(\bs) - \mu_Z(\bs)\rvert \geq \epsilon) \leq 2 - 2\Phi\left(\tfrac{\epsilon L^\nu}{  \sigma_Z
\sqrt{\frac{1}{\nu}+\frac{1}{L(2\nu+1)} }}\right). \]
\end{prop}

\begin{proof}
Let $C(\bs,\bs')$ and $C^L(\bs,\bs')$ denote the covariance function of $\mu_Z$ and $\mu_Z^L,$ respectively. The series expansion of the covariance function is bounded uniformly, with
    \begin{align*}
    \sup_{\bs,\bs' \in \mathbb{S}^2}\lvert C(\bs,\bs') - C^L(\bs,\bs')\lvert &=\sup_{\mathbf{s},\mathbf{s'}\in\mathbb{S}^2} \sigma^2_Z\left \lvert \sum_{l=L+1}^\infty (2l + 1)S_l P_l({\bf s} \cdot {\bf s'}) \right \rvert \\
    &= \sigma^2_Z\sum_{l=L+1}^\infty (2l+1)(\kappa^2 + (l+1)l)^{-(\nu+1)}\\
    % &\leq \sigma^2_Z\sum_{l=L+1}^\infty (2l+1)(\kappa^2 + l^2)^{-(\nu+1)}\\
    &\leq \sigma^2_Z\sum_{l=L+1}^\infty (2l+1)l^{-2(\nu+1)}\\
    &=  \sigma^2_Z \sum_{x=1}^\infty (2(x+L)+1)(x+L)^{-2(\nu+1)}\\
    &\leq \sigma^2_Z\int_0^\infty (2(x+L)+1)(x+L)^{-2(\nu+1)}dx\\
    &= \sigma^2_Z\left(\tfrac{1}{\nu} + \tfrac{1}{L(2\nu+1)}\right)L^{-2\nu}.
\end{align*}
Setting $\bs = \bs',$ the above result and Markov inequality implies 
\begin{align*}
    \mathbb{P}(\lvert\mu_Z^L(\bs)-\mu_Z(\bs)\rvert > \epsilon) &\leq
    \frac{\sigma^2_Z}{\epsilon^2}\left(\tfrac{1}{\nu} + \tfrac{1}{2\nu+1}\right)L^{-2\nu}.
\end{align*}
The series $\sum_{L=1}^\infty L^{-2\nu}$ converges for $\nu > 1/2,$ which proves the almost sure convergence of $\mu_Z^L(\bs) \to \mu_Z(\bs)$ by the Borel-Cantelli lemma.

A tighter bound given in the proposition can be obtained by considering the Gaussian structure of $\mu_Z$ and noting that $\mu_Z^L(\bs) - \mu_Z(\bs)$ follows a normal distribution with mean $0$ and variance bounded by $\sigma^2_Z\left(\tfrac{1}{\nu} + \tfrac{1}{L(2\nu+1)}\right)L^{-2\nu}.$ 
\end{proof}

\begin{proof}[Proof of Theorem \ref{thm:cp_trunc}]
    Consider the case where $a=0.$ $\tau^L(\bs) = \tau(\bs)$ if and only if $Z^L(\bs)$ and $Z(\bs)$ fall in the same bin defined by $\gamma$'s, Then, 
    \begin{align*}
        & \mathbb{P}(\lvert \tau^L(\bs)-\tau(\bs)\rvert = 0)\\
         \geq & \mathbb{P}\left(\lvert \mu_Z^L(\bs)-\mu_Z(\bs)\rvert \leq \min_{m}\lvert \gamma_m - Z^L(\bs)\rvert \right)\\
        = & \sum_{k=1}^M \mathbb{P}\left(\lvert \mu_Z^L(\bs)-\mu_Z(\bs)\rvert \leq \min_{m}\lvert \gamma_m - Z^L(\bs)\rvert  \;\middle|\; \tau^L(\bs) = k\right)\mathbb{P}\left(\tau^L(\bs) = k\right)\\
        = & \sum_{k=1}^M \mathbb{P}\left(\lvert \mu_Z^L(\bs)-\mu_Z(\bs)\rvert \leq \Delta_{k,0}(Z^L(\bs))  \;\middle|\; \gamma_{k-1}\leq Z^L(\bs) < \gamma_k\right)\\
        &\hspace{1.cm}\times \mathbb{P}\left(\gamma_{k-1}\leq Z^L(\bs) < \gamma_k\right)\\
        = & \sum_{k=1}^M \int_z \mathbb{P}\left(\lvert \mu_Z^L(\bs)-\mu_Z(\bs)\rvert \leq \Delta_{k,0}(z)  \;\middle|\;Z^L(\bs) =z\right)f_{Z^L\mid Z^L\in (\gamma_{k-1},\gamma_k]}(z)dz\\
        &\hspace{1.cm}\times \mathbb{P}\left(\gamma_{k-1}\leq Z^L(\bs) < \gamma_k\right)\\
        = & \sum_{k=1}^M \mathbb{E}\left[ \mathbb{P}\left(\lvert \mu_Z^L(\bs)-\mu_Z(\bs)\rvert \leq \Delta_{k,0}(Z^L(\bs))  \;\middle|\; \gamma_{k-1}\leq Z^L(\bs) < \gamma_k\right)\right]\\
        &\hspace{1.cm}\times \mathbb{P}\left(\gamma_{k-1}\leq Z^L(\bs) < \gamma_k\right).
    \end{align*}
    where the last equality follows from independence of $\mu_Z^L(\bs)-\mu_Z(\bs)$ and $Z^L(\bs).$ The desired result follows from Proposition~\ref{prop:conv}. 
\end{proof}

\begin{cor}
Let $U_a$ denote the upper bound for $\mathbb{P}(\lvert \tau^L(\bs)-\tau(\bs) \rvert \leq a)$ given in Theorem~\ref{thm:cp_trunc}. Then, the expected $L^1$ error is bounded by
\begin{equation}
    \mathbb{E}[\lVert \tau^L - \tau \rVert_{L^1(\mathbb{S}^2)}] \leq 4\pi \left(M-1- \sum_{a=0}^{M-1}U_a \right).
\end{equation}
\end{cor}
\begin{proof}
    \begin{align*}
        \mathbb{E}[\lVert \tau^L - \tau \rVert_{L^1(\mathbb{S}^2)}] &= \mathbb{E}\left(\sum_{i=1}^N w(\bs_i)\lvert\tau^L(\bs_i)-\tau(\bs_i)\rvert\right)\\
        &= \sum_{i=1}^N w(\bs_i) \sum_{a=0}^{M-1}a\mathbb{P}(\lvert \tau^L(\bs)-\tau(\bs)\rvert = a )\\
        &= 4\pi \left((M-1) - \sum_{a=0}^{M-1}\mathbb{P}(\lvert \tau^L(\bs)-\tau(\bs)\rvert \leq a)\right)\\
        & \leq 4\pi \left((M-1) - \sum_{a=0}^{M-1}U_a\right).
    \end{align*}
\end{proof}
Since the distribution of $Z^L(\bs)$ is known, the bounds can be obtained numerically. In Table~\ref{tab:trunc}, we provide the worst-case expected MAE (defined as $\tfrac{\mathbb{E}[\lVert \tau^L - \tau \rVert_{L^1(\mathbb{S}^2)}]}{4\pi}$) for $180 \times 360$ grid under three different scenarios: (i) $\gamma_k = \sqrt{v_Z+1}\Phi^{-1}(\pi_k),$ $k=1,\ldots,M-1,$ where $\pi_k$'s were chosen based on the posterior distribution of the changepoints in a real world dataset from Section~\ref{sec:data}, (ii) $\gamma_k = \sqrt{v_Z+1}\Phi^{-1}(\tfrac{k}{M})$ (equal-probability), and (iii) $\gamma_k = -B + k(\tfrac{2B}{M}),$ where the endpoints $\pm B$ were chosen by simulating $Z^L(\bs)$ 10,000 times and averaging the spread (equal-distance).

%\clearpage % flush earlier floats and start fresh
\begin{table}[p]
\centering
\caption{Worst-case expected MAE}
\label{tab:trunc}

\begin{singlespace}
\scriptsize
\setlength{\tabcolsep}{3pt}
\renewcommand{\arraystretch}{1.0}

\begin{minipage}[t]{0.32\textwidth}
\centering
\caption*{(i) real data}
\begin{tabular}{cccccc}
\toprule
$M$ & \begin{tabular}{c}Observed\\Categories\end{tabular} & $v_Z$ & $\kappa$ & $\nu$ & MAE \\
\midrule
60 & 8 & 1  & 5 & 1 & 0.2221 \\
60 & 8 & 1  & 5 & 2 & 0.0122 \\
60 & 8 & 1  & 5 & 3 & 0.0000 \\
60 & 8 & 5  & 5 & 1 & 0.2605 \\
60 & 8 & 5  & 5 & 2 & 0.0143 \\
60 & 8 & 5  & 5 & 3 & 0.0000 \\
60 & 8 & 10 & 5 & 1 & 0.2655 \\
60 & 8 & 10 & 5 & 2 & 0.0146 \\
60 & 8 & 10 & 5 & 3 & 0.0002 \\
\bottomrule
\end{tabular}
\end{minipage}\hfill
\begin{minipage}[t]{0.32\textwidth}
\centering
\caption*{(i) equal-probability}
\begin{tabular}{cccccc}
\toprule
$M$ & \begin{tabular}{c}Observed\\Categories\end{tabular} & $v_Z$ & $\kappa$ & $\nu$ & MAE \\
\midrule
10  & 10  & 1  & 5 & 1 & 0.2913 \\
10  & 10  & 1  & 5 & 2 & 0.0160 \\
10  & 10  & 1  & 5 & 3 & 0.0001 \\
10  & 10  & 5  & 5 & 1 & 0.4014 \\
10  & 10  & 5  & 5 & 2 & 0.0221 \\
10  & 10  & 5  & 5 & 3 & 0.0005 \\
10  & 10  & 10 & 5 & 1 & 0.4258 \\
10  & 10  & 10 & 5 & 2 & 0.0234 \\
10  & 10  & 10 & 5 & 3 & 0.0009 \\
\addlinespace
50  & 50  & 1  & 5 & 1 & 1.8990 \\
50  & 50  & 1  & 5 & 2 & 0.1171 \\
50  & 50  & 1  & 5 & 3 & 0.0032 \\
50  & 50  & 5  & 5 & 1 & 2.4789 \\
50  & 50  & 5  & 5 & 2 & 0.1608 \\
50  & 50  & 5  & 5 & 3 & 0.0048 \\
50  & 50  & 10 & 5 & 1 & 2.6015 \\
50  & 50  & 10 & 5 & 2 & 0.1705 \\
50  & 50  & 10 & 5 & 3 & 0.0071 \\
\addlinespace
100 & 100 & 1  & 5 & 1 & 3.8651 \\
100 & 100 & 1  & 5 & 2 & 0.2644 \\
100 & 100 & 1  & 5 & 3 & 0.0079 \\
100 & 100 & 5  & 5 & 1 & 5.0110 \\
100 & 100 & 5  & 5 & 2 & 0.3634 \\
100 & 100 & 5  & 5 & 3 & 0.0119 \\
100 & 100 & 10 & 5 & 1 & 5.2532 \\
100 & 100 & 10 & 5 & 2 & 0.3853 \\
100 & 100 & 10 & 5 & 3 & 0.0165 \\
\bottomrule
\end{tabular}
\end{minipage}\hfill
\begin{minipage}[t]{0.32\textwidth}
\centering
\caption*{(iii) equal-distance}
\begin{tabular}{cccccc}
\toprule
$M$ & \begin{tabular}{c}Observed\\Categories\end{tabular} & $v_Z$ & $\kappa$ & $\nu$ & MAE \\
\midrule
10  & 10  & 1  & 3   & 1 & 0.0800 \\
10  & 10  & 1  & 5   & 1 & 0.1248 \\
10  & 10  & 1  & 100 & 1 & 1.2797 \\
10  & 10  & 5  & 3   & 1 & 0.1343 \\
10  & 10  & 5  & 5   & 1 & 0.1847 \\
10  & 10  & 5  & 100 & 1 & 1.4902 \\
10  & 10  & 10 & 3   & 1 & 0.1595 \\
10  & 10  & 10 & 5   & 1 & 0.2082 \\
10  & 10  & 10 & 100 & 1 & 1.5306 \\
\addlinespace
50  & 50  & 1  & 3   & 1 & 0.5987 \\
50  & 50  & 1  & 5   & 1 & 0.9308 \\
50  & 50  & 1  & 100 & 1 & 6.5040 \\
50  & 50  & 5  & 3   & 1 & 0.9853 \\
50  & 50  & 5  & 5   & 1 & 1.3303 \\
50  & 50  & 5  & 100 & 1 & 7.5096 \\
50  & 50  & 10 & 3   & 1 & 1.1518 \\
50  & 50  & 10 & 5   & 1 & 1.4732 \\
50  & 50  & 10 & 100 & 1 & 7.7033 \\
\addlinespace
100 & 100 & 1  & 3   & 1 & 1.3412 \\
100 & 100 & 1  & 5   & 1 & 1.9890 \\
100 & 100 & 1  & 100 & 1 & 13.0142 \\
100 & 100 & 5  & 3   & 1 & 2.0956 \\
100 & 100 & 5  & 5   & 1 & 2.7433 \\
100 & 100 & 5  & 100 & 1 & 15.0228 \\
100 & 100 & 10 & 3   & 1 & 2.4115 \\
100 & 100 & 10 & 5   & 1 & 3.0177 \\
100 & 100 & 10 & 100 & 1 & 15.4097 \\
\bottomrule
\end{tabular}
\end{minipage}

\end{singlespace}
\end{table}
\clearpage

\section{Space-time separability}\label{app:separability}
% Separability in space and time is a widely used and convenient simplification for spatiotemporal covariance models. However, it can degrade statistical performance \citep{li2007nonparametric}. In what follows, we examine the covariance structure of the model developed in Section~\ref{sec:st}, characterize its affine relationship with the Matérn class, and investigate the impact of space–time separability on changepoint analysis.
% While assuming separability in space and time is a common and convenient simplification in spatiotemporal covariance modeling, it can come at the cost of reduced statistical performance \citep{li2007nonparametric}. To better understand these trade-offs, 
We analyze the covariance structure of the model introduced in Section~\ref{sec:st}, explore its affine relationship with the Matérn class, and investigate how the assumption of separability affects changepoint analysis.

The cross covariance of $U(s,t)$ is given by
\begin{align*}
    \text{Cov}(U(s,t),U(s',t')) &= \text{Cov}(\cU(t)(s),\cU(t')(s'))\\
    &= \text{Cov}\left(\sum_{l=0}^\infty\sum_{m=-l}^l \hat{\cU}_{lm}(t)\psi_{lm}(s),\sum_{l=0}^\infty\sum_{m=-l}^l \hat{\cU}_{lm}(t')\psi_{lm}(s')\right)\\
    &= \sum_{l,l'}\sum_{m,m'} \psi_{lm}(s)\psi_{l'm'}(s') \text{Cov}(\hat{\cU}_{lm}(t),\hat{\cU}_{l',m'}(t'))\\
    &= \sum_{l}\sum_{m} \psi_{lm}(s)\psi_{lm}(s') \text{Cov}(\hat{\cU}_{lm}(t),\hat{\cU}_{l,m}(t')).
\end{align*}
% where the last equality follows from the independence of $\hat{\cU}_{lm}(t),\hat{\cU}_{l',m'}(t)$ for $l\neq l'$ or $m\neq m'.$
Let $u$ denote the angle between $\bs$ and $\bs'$ and assume $t' = t+h,$ $h>0.$ By It\^o's isometry and spherical harmonics addition theorem, the above expression becomes 
\begin{align}\label{eq:crosscov}
    \cov(U(s,t),U(s',t')) %&= \sigma^2_Q \sum_{l=0}^\infty \sum_{m=-l}^l \psi_{lm}(s)\psi_{lm}(s') \frac{(\kappa^2 + l(l+1))^{-(\nu+1)}}{\xi_r + \xi_d l(l+1)}e^{-(\xi_r + \xi_d l(l+1))h} \\
    &= \frac{\sigma^2_Q}{4\pi} \sum_{l=0}^\infty (2l+1)P_l(\cos (u))\frac{(\kappa^2 + l(l+1))^{-(\nu+1)}}{\xi_r + \xi_d l(l+1)}e^{-(\xi_r + \xi_d l(l+1))h}. 
\end{align}

Denote by $C_{s,t}(u,h) := \cov(U(s,t),U(s',t'))$ the space-time cross covariance, $C_s(u) := \cov(U(\bs,t),U(\bs',t))$ the marginal spatial covariance, and $C_t(h) := \cov(U(\bs,t),U(\bs,t'))$ the marginal temporal covariance. Two special cases of  $C_{s,t}(u,h)$ merit particular attention, as they highlight interesting properties of the covariance structure. First, when $\xi_r/\xi_d = \kappa^2,$ its marginal spatial covariance reduces to 
\begin{align*}
    C_s(u) &\propto \sum_{l=0}^\infty \sum_{m=-l}^l(\kappa^2 + l(l+1))^{-(\nu+2)} \psi_{lm}(s)\psi_{lm}(s'),
\end{align*}
which coincides with the covariance of the spatial process defined by the Whittle-Mat\'ern SPDE \eqref{eq:SPDE} with inverse range parameter $\kappa$ and smoothness $\nu + 1.$ This can be seen as the reparametrization of the diffusion model in \cite{lindgren2020diffusion}. Second, when $\xi_d = 0$, the temporal decay rate is independent of $l$ and the cross covariance in \eqref{eq:crosscov} factors as
\begin{align}\label{eq:sepcov}
     C_{s,t}(u,h) 
    &= \frac{\sigma^2_Q}{\xi_r} e^{-\xi_rh}\sum_{l=0}^\infty \sum_{m=-l}^l(\kappa^2 + l(l+1))^{-(\nu+1)} \psi_{lm}(s)\psi_{lm}(s')\\
    &\propto C_t(h)C_s(u), \nonumber
\end{align}
implying a space-time separable structure. In this case, the marginal spatial covariance $C_s(u)$ coincides again with that of the spatial process defined by the Whittle–Matérn SPDE \eqref{eq:SPDE} with inverse range $\kappa$ and smoothness $\nu$, and the marginal temporal covariance $C_t(h) \propto (e^{-\xi_r})^h$ corresponds to an AR(1) process with autoregressive coefficient $e^{-\xi_r}$. This suggests that the diffusivity parameter $\xi_d$ governs the interaction between spatial and temporal dependencies of $U$, motivating the study of space-time separability.

To quantify deviations from separability for general $\xi_d$, we first define the correlation functions as $\rho_{st}(\cdot) = C_{s,t}(\cdot) / C_{s,t}(\mathbf{0}, 0)$, $\rho_s(\cdot) = C_s(\cdot)/C_s(\mathbf{0})$, and $\rho_t(\cdot) = C_t(\cdot)/C_t(0)$. We then define the supremum norm
\[c_{sep} := \max_{u,h}|\rho_{st}(u,h) - \rho_s(u)\rho_t(h) |\]
as a measure of non-separability. Note that $c_{sep} = 0$ if and only if the covariance is separable, and a larger value of $c_{sep}$ indicates higher degree of non-separability.

Let $f_l(\xi_r,\xi_d) := (2l+1)\frac{(\kappa^2+l(l+1))^{-(\nu+1)}}{\xi_r + \xi_d l(l+1)}$.  A direct calculation shows that   
\begin{align*}
    c_{sep}(\xi_r,\xi_d) &= \max_{u,h} \left \lvert e^{-\xi_rh} \sum_{l\neq l'}f_l(\xi_r,\xi_d)f_{l'}(\xi_r,\xi_d)P_l\left(\cos (u)\right)\left(e^{-\xi_dl'(l'+1)h} - e^{-\xi_dl(l+1)h}\right) \right \rvert.
\end{align*}
It can be seen that $c_{sep} = 0$ when $\xi_d = 0$ and $c_{sep} \to 0$ as $\xi_d \to \infty$, which corresponds to the case of vanishing marginal temporal correlation. Figure \ref{fig:csep} plots the $c_{sep}$ as a function of $\xi_d$, with each panel corresponding to different values of $\kappa$ and $\nu.$ Within each panel, the curves for varing values of $\xi_r$ are distinguished by color. The peak of each curve represents the parameter combination that yields the "maximally non-separable" covariance, with each peak associated with different values of marginal temporal and spatial correlation. 

\begin{figure}[!ht]
\centering
\includegraphics[width=0.8\textwidth]{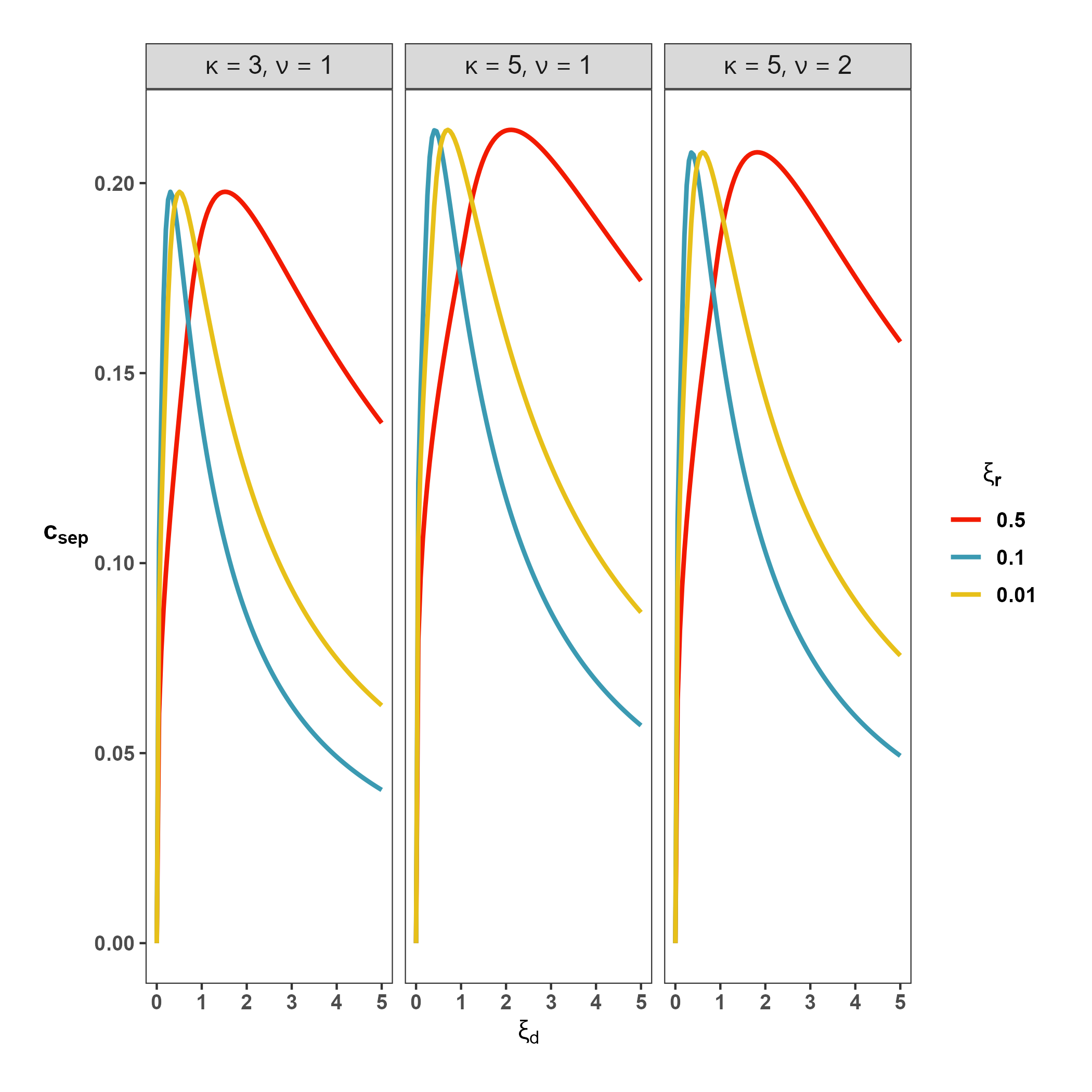}
\caption{\label{fig:csep} Plot of $c_{sep}$ in the y-axis vs. $\xi_d$ in the x-axis for different values of $\kappa,\nu.$ Each curve represent different values of $\xi_r,$ distinguished by color. } 
\end{figure}

\section{Sparsity of the spherical harmonic error operator}\label{app:sh_proof}
\subsection{Spherical Harmonics}
\begin{defn} The real spherical harmonics are given by
    \begin{align*}
        \psi_{lm}(\theta,\phi) &= \begin{cases} \bar{P}_{lm}(\cos \theta) \cos (m\phi), & m \geq 0\\ \bar{P}_{l|m|}(\cos \theta) \sin(|m|\phi), &m < 0, \end{cases}\\
    \intertext{where $\bar{P}_{lm}$ are the orthonormalized associated Legendre polynomials}
    \bar{P}_{lm}(\mu) &= \sqrt{\frac{(2-\delta_{0m})(2l+1)}{4\pi} \frac{(l-m)!}{(l+m)!}} (1-\mu^2)^{m/2}\frac{d^m}{d\mu^m}P_l(\mu)
    \end{align*}
and $P_l$ are the standard Legendre polynomials.
\end{defn}
For fixed value of $m,$ $\bar{P}_{lm}$ are orthogonal:
\begin{align*}
    \int_{-1}^1 \bar{P}_{lm}(\mu)\bar{P}_{l'm}(\mu)d\mu = \frac{(2-\delta_{0m})}{2\pi}\delta_{ll'},
\end{align*}
and the spherical harmonics $\psi_{lm}$ are orthonormal over both $l$ and $m$:
\begin{align*}
    \int_{\bs \in \mathbb{S}^2} \psi_{lm}(\bs)\psi_{l'm'}(\bs) d\bs &= \delta_{mm'}\delta_{ll'}.
\end{align*}

\subsection{Proof of Theorem~\ref{thm:sparse}}

\begin{proof}[Proof of Theorem~\ref{thm:sparse}]
Let $A(m,\phi) = \begin{cases} \cos (m\phi) &m\ge 0, \\ \sin(|m|\phi) & m <0 \end{cases}.$ The entries of $\boldsymbol{\Psi}' \mathbf{D}^2_w\boldsymbol{\Psi}$ are given by 
\begin{align}\label{eq:sparse_entries}
& \sum_{j=0}^{2K-1}\sum_{i=0}^{K-1} \psi_{lm}(\theta_i,\phi_j)\psi_{l'm'}(\theta_i,\phi_j)w(\theta_i)^2  \nonumber \\ 
= & \sum_{j=0}^{2K-1} A(m,\tfrac{\pi j}{K})A(m',\tfrac{\pi j}{K}) \sum_{j=0}^{K-1}\bar{P}_{lm}(\cos \tfrac{\pi i}{K})\bar{P}_{l'm'}(\cos \tfrac{\pi i}{K})w(\tfrac{\pi i}{K})^2 \nonumber \\
= & \begin{dcases} \frac{2K}{2-\delta_{0m}} \sum_{i=0}^{K-1}\bar{P}_{lm}(\cos \tfrac{\pi i}{K})\bar{P}_{l'm'}(\cos \tfrac{\pi i}{K})w(\tfrac{\pi i}{K})^2 &\quad \text{ if }m = m'\\
0 & \phantom{ if } \quad m\neq m',
\end{dcases}
\end{align}
where the Driscoll-Healy weights are given by \citep{driscoll1994computing}
\[w(\theta) = \frac{2\sqrt{2}}{K}\sin \theta \sum_{j=0}^{\nicefrac{K}{2}-1} \frac{1}{2j+1}\sin ((2j+1)\theta).\]
Consider the case where $m=m'$. We have
\begin{align*}
    & \sum_{i=0}^K\bar{P}_{lm}(\cos \tfrac{\pi i}{K})\bar{P}_{l'm'}(\cos \tfrac{\pi i}{K})w(\tfrac{\pi i}{K})^2\\
    = & \frac{1}{2}\sum_{i=-K}^K\bar{P}_{lm}(\cos \tfrac{\pi i}{K})\bar{P}_{l'm'}(\cos \tfrac{\pi i}{K})w(\tfrac{\pi i}{K})^2 d\theta\\
    = & \frac{K}{2\pi} \int_{-\pi}^\pi \bar{P}_{lm}(\cos \theta)\bar{P}_{l'm'}(\cos \theta)\sin^2\theta w(\theta)^2 d\theta \\
    = & \frac{\pi}{4K} \int_{-\pi}^\pi \bar{P}_{lm}(\cos \theta)\bar{P}_{l'm'}(\cos \theta)\sin^2\theta \left(\frac{4}{\pi}\sum_{j=0}^{\nicefrac{K}{2}-1}\frac{\sin((2j+1)\theta)}{2j+1} \right)^2 d\theta\\
    = & \frac{\pi}{2K} \int_0^\pi \bar{P}_{lm}(\cos \theta)\bar{P}_{l'm'}(\cos \theta)\sin^2\theta \left(\frac{4}{\pi}\sum_{j=0}^{\infty}\frac{\sin((2j+1)\theta)}{2j+1} \right)^2 d\theta\\
    = & \frac{\pi}{2K} \int_0^\pi \bar{P}_{lm}(\cos \theta)\bar{P}_{l'm'}(\cos \theta)\sin^2\theta d\theta.
\end{align*}
Using the recurrence formula
\[\sqrt{1-x^2} P_{lm}(x) = \frac{-1}{2l+1}(P_{l+1, m+1}(x) - P_{l-1,m-1}(x)), \]
the integral can be expressed as
\begin{align*}
    \int_0^\pi \bar{P}_{lm}(\cos \theta)\bar{P}_{l'm'}(\cos \theta)\sin^2\theta d\theta &= \tfrac{1}{2l'+1}\int_{-1}^1 \bar{P}_{lm}(x)\bar{P}_{l'-1,m+1}(x)dx - \int_{-1}^1 \bar{P}_{lm}(x)\bar{P}_{l'+1,m+1}(x)dx
\end{align*}
The overlap integral of two associated Legendre polynomials is given by (\cite{mavromatis1999single}, \cite{dong2002overlap})
\begin{align*}
    \int_{-1}^1 P_{l_1m_1}(x)P_{l_2m_2}(x)dx &= A(l_1,m_1,l_2,m_2)\sum_k B(|m_2-m_1|,k)(2k+1)\begingroup \renewcommand{\arraystretch}{0.5}
\begin{pmatrix}
l_1 & l_2 & k \\
0 & 0 & 0
\end{pmatrix}
\endgroup\\
    &\hspace{3.5cm} \times \begingroup
\renewcommand{\arraystretch}{0.5}
\begin{pmatrix}
l_1 & l_2 & k \\
-m_1 & m_2 & m_1-m_2
\end{pmatrix}
\endgroup,
\end{align*}
where 
\begin{align*}
    A(l_1,m_1,l_2,m_2) &= \tfrac{(-1)^{m_1}}{4\pi}|m_2-m_1|2^{|m_2-m_1|-2}\sqrt{(2-\delta_{0m_1})(2-\delta_{0m_2})(2l_1+1)(2l_2+1)},\\
    B(|m_2-m_1|,k) &= (1+(-1)^{k-|m_2-m_1|})\sqrt{\tfrac{(k-|m_2-m_1|)!}{(k+|m_2-m_1|)!}}\tfrac{\Gamma(\tfrac{k}{2})\Gamma(\tfrac{k+|m_2-m_1|+1}{2})}{(\tfrac{k-|m_2-m_1|}{2})! \Gamma(\tfrac{k+3}{2})},
\end{align*}
$\begingroup
\renewcommand{\arraystretch}{0.5}
\begin{pmatrix}
l_1 & l_2 & k \\
0 & 0 & 0
\end{pmatrix}
\endgroup$ and $\begingroup
\renewcommand{\arraystretch}{0.5}
\begin{pmatrix}
l_1 & l_2 & k \\
-m_1 & m_2 & m_1-m_2
\end{pmatrix}
\endgroup$ 
are 3-j symbols, and $k$ satisfies $\lvert l_1-l_2\rvert \leq k \leq l_1+l_2,$ $k \geq \lvert m_1-m_2\rvert$, $k+l_1+l_2$ is even, and $k + m_2-m_1$ is even. Setting $(l_1,m_1,l_2,m_2) = (l,m,l'\pm 1,m+1),$ the last two conditions are satisfied only when $k$ is odd and $l'+l$ is even. Thus, \eqref{eq:sparse_entries} is non-zero only when $m=m'$ and $l'+l$ is even.
\end{proof}

\end{document}